\documentclass[journal,draftcls,onecolumn,oneside]{IEEEtran}
\ifCLASSINFOpdf
   \usepackage[pdftex]{graphicx}
   \DeclareGraphicsExtensions{.pdf,.jpeg,.png}
\else
   \usepackage{graphicx}
   \DeclareGraphicsExtensions{.eps}
\fi
\usepackage{subfigure}
\usepackage{arydshln}
\usepackage[cmex10]{amsmath}
\usepackage{amssymb}
\usepackage{bm}
\usepackage{mathrsfs}
\usepackage{caption}
\usepackage[T1]{fontenc}

\DeclareMathOperator{\Span}{span}
\DeclareMathOperator{\Diag}{diag}
\DeclareMathOperator{\Vol}{vol}
\DeclareMathOperator{\Rank}{rank}
\DeclareMathOperator{\Trace}{Tr}

\DeclareMathOperator{\Dim}{dim}
\DeclareMathOperator{\Corr}{corr}

\newtheorem{Def}{\textbf{Definition}}
\newtheorem{Problem}{\textbf{Problem}}
\newtheorem{Theorem}{\textbf{Theorem}}
\newtheorem{Lemma}{\textbf{Lemma}}

\begin{document}

\title{The Volume-Correlation Subspace Detector}
\author{Hailong Shi, Hao Zhang and Xiqin Wang %
\thanks{Intelligent Sensing Lab, Department of Electronic Engineering, Tsinghua University, Beijing, E-mail: shl06@mails.tsinghua.edu.cn, haozhang@mail.tsinghua.edu.cn, wangxq\_ee@tsinghua.edu.cn.}}

\maketitle

\begin{abstract}
Detecting the presence of subspace signals with unknown clutter (or interference) is a widely known difficult problem encountered in various signal processing applications. Traditional methods fails to solve this problem because they require knowledge of clutter subspace, which has to be learned or estimated beforehand. In this paper, we propose a novel detector, named volume-correlation subspace detector, that can detect signal from clutter without any knowledge of clutter subspace. This detector effectively makes use of the hidden geometrical connection between the known target signal subspace to be detected and the subspace constructed from sampled data to ascertain the existence of target signal. It is derived based upon a mathematical tool, which basically calculates volume of parallelotope in high-dimensional linear space. Theoretical analysis show that while the proposed detector is detecting the known target signal, the unknown clutter signal can be explored and eliminated simultaneously. This advantage is called "detecting while learning", and implies perfect performance of this detector in the clutter environment. Numerical simulation validated our conclusion.
\end{abstract}

\section{Introduction}

\par
We consider the following problem widely existing in communication, radar, sonar and other fields of signal detection and processing: How to effectively detect a target signal buried in clutter lying in an UNKNOWN low rank subspace and random noise. The noticeable point here is that the clutter subspace is unknown at the receiver. Without what is called training data, we could not sample directly in this clutter subspace and get enough information to eliminate its influence on target signal detection.

As a matter of fact, detecting target signal in certain signal subspace has been considered by several researchers and various schemes has been proposed.
Denote $\bm{H}$ by Hilbert space, which is the basic signal space, then the problem of detecting signal in subspace could be formulated roughly as follows.

\begin{Problem}\label{Prob1}
Suppose $\bm{H}_S\subset\bm{H}$ be the known target signal subspace with dimension $d$,
\begin{equation}
\bm{H}_S=\Span\{\bm s_1,\bm s_2,\cdots,\bm s_{d}\}
\end{equation}
where $\bm s_k\in\bm{H}, k=1,\cdots,d, d< n$ are KNOWN signal vectors. Given the sampled data $\bm y\in\bm{H}$ contaminated by Gaussian random noise, could we determine whether or not $\bm y$ lies in $\bm{H}_S$?
\end{Problem}

The well-known projection method has long been regarded as the basic step for these kinds of detectors \cite{KayV1,KayV2}. Indeed, to solve Problem \ref{Prob1}, we firstly project the sampled data $\bm y$ onto subspace $\bm H_S$, then we use the energy detector to make the decision. Actually, the conventional optimal detector for single known signal vector, i.e., the matched filter, is also a special kind of this projection based detector. Because the known signal vector indeed represents a one-dimensional subspace, and the 
matched filter can be regarded as essentially sequential projections of received data on the time-shift versions of target signal. On the other hand, least square estimation and its alternatives also involve projection of raw data on the subspace spanned by several prescribed signal vectors. The optimality of linear least square method as the linear estimator could be sufficiently guaranteed by Gauss-Markov theorem in statistical inference. It is clear that detection, being the counterpart of estimation, of signal in known subspace is closely contacted with projection operation in linear space. In fact, more complicated problem involving clutter or interference has been solved by projection-based methods, such problem is formulated as follows.

\begin{Problem}
In addition to $\bm{H}_S$, suppose $\bm{H}_C\subset\bm{H}$ be the KNOWN clutter subspace with dimension $m$.
\begin{equation}
\bm{H}_C=\Span\{\bm c_1,\bm c_2,\cdots,\bm c_m\}
\end{equation}
where $\bm c_k\in\bm{H}, k=1,\cdots,m,m<n$ are KNOWN clutter (interference) vectors. Given the sampled data $\bm y\in\bm{H}$, could we determine whether or not $\bm y$ lies in $\bm{H}_S\oplus\bm{H}_C$ but not entirely in $\bm{H}_C$ from the contamination of random noise? In other words, whether $\bm y$ satisfies
\begin{equation}
\bm y=\bm s+ \bm c + \bm w, 
\end{equation}
or
\begin{equation}
\bm y=\bm c + \bm w, 
\end{equation}
where $\bm s\in\bm{H}_S,\ \bm c\in\bm{H}_C,\ \bm s, \bm c\neq0$ and $\bm w$ is the Gaussian random noise.
\end{Problem}

Unlike Problem 1, Problem 2 can not be solved by the simple projection and energy detector mentioned above. The reason is that clutter subspace $\bm{H}_C$ is not necessarily orthogonal to signal subspace $\bm{H}_S$ in general (they are even alike in hostile environment). In practical scenario the power of clutter is often much higher than that of target signal. So it is hard to tell the existence of signal component from sample data $\bm y$ only based on the argument that the energy of its projection on $\bm{H}_S$ is relatively large. The most remarkable approach to the problem is the Matched Subspace Detector firstly proposed by L. Scharf et.al. \cite{scharf1994matched}. It is essentially a "two-folds" projection: Firstly, sampled data $\bm y$ was projected on orthogonal complement of $\bm{H}_C$ to eliminate the influence of clutter thoroughly; Secondly, the result of first projection was further projected onto the part of $\bm{H}_S$ that is orthogonal to $\bm{H}_C$, and then an energy detector was applied to infer the existence of signal component in $\bm y$. Denoting $\mathbb{P}_{\bm H}$ and $\mathbb{P}_{\bm H}^\perp$ by the projection operators on subspace $\bm{H}$ and its orthogonal complement $\bm{H}^\perp$ respectively, the matched subspace detector could be written as
\begin{equation}
t(\bm y) = \frac{1}{\sigma^2}\|\mathbb{P}_{\mathbb{P}_{\bm H_C}^{\perp}\bm{ H}_S}\mathbb{P}_{{\bm H}_C}^{\perp}\bm y\|_2^2,
\end{equation}

\par
Although tremendous variations and applications of matched subspace detector has appeared \cite{scharf1996adaptive}\cite{kraut1999cfar}\cite{kraut2001adaptive}\cite{desai2003robust}\cite{mcwhorter2003matched}, it should be noticed that the key precondition for the success of matched subspace detector is the clutter subspace $\bm{H}_C$ must be KNOWN beforehand. It is seldom satisfied in practice, especially in radar, reconnaissance, mobile communication and underwater signal processing. Hence the problem we encounter is actually like this:

\begin{Problem}\label{Pro3}
Let $\bm{H}_S$ be KNOWN target signal subspace and $\bm{H}_C$ be UNKNOWN clutter subspace, given the sampled data $\bm y\in\bm{H}$ contaminated by Guassian noise, if we assume $\dim(\bm H_S \bigcap \bm H_C)=0$, could we determine whether $\bm y$ contains signal that lies in $\bm{H}_S\oplus\bm{H}_C$ or entirely in $\bm{H}_C$? In other words, whether $\bm y$ satisfies
\begin{equation}
\bm y=\bm s+\bm c + \bm w,\label{signalplusclutter} 
\end{equation}
or
\begin{equation}
\bm y=\bm c + \bm w, 
\end{equation}
where $\bm s\in\bm{H}_S,\ \bm c\in\bm{H}_C,\ \bm s, \bm c\neq0$ and $\bm w$ is the Gaussian random noise.
\end{Problem}

Because of the unknownness of clutter subspace, projection-based detectors can not be constructed explicitly. However, the structure of clutter subspace could be explored by successively sampling in it to gain the information of the basis for subspace. Because the generic property of randomly sampling in linear space ensures the linear independence of sampled vectors, the clutter subspace could be 'reconstructed' by multiple samples. But it should be noted that, the information of target signal is mixed intimately with the clutter in the samples, like the case shown in (\ref{signalplusclutter}). In this case, it is impossible to separated the 'pure' clutter signal $\bm c$ from target signal $\bm s$ so that the basis for clutter subspace could be extracted alone. In other words, with multiple sampled vectors of $\bm y$ that satisfies the generic property, we cannot determine whether the sample subspace these $\bm y$ span is $\bm H_S\oplus \bm H_C$ or just $\bm H_C$, this issue is the core difficulty for detecting target signal against structured deterministic clutter. 

It is interesting to make a comparison of our problem to the problem of detection in random noise without clutter. Both two detection problems have similar formulations. In fact, there are two hypothesis, $H_0$ and $H_1$, which are
\begin{align}\label{SubspaceRel}
H_1:&\bm y = \bm s + \bm n, \nonumber\\
H_0:&\bm y = \bm n,
\end{align}
The only difference is that $\bm n$ in traditional detection problems appears to be only random noise and is described by a certain probability distribution. On the contrary, $\bm n$ in our problem includes both random noise and a vector lying within certain deterministic and unknown low-dimensional subspace. Detectors based on Likelihood Ratio Test have been proved to be optimal under the probabilistic assumption of $\bm n$ and projection is the natural consequence of likelihood ratio test in the case of Gaussian distributed noise. However, it can't be applied directly in our problem for the deterministic and unknown clutter subspace structure. Hence, more effective detector which could take fully advantage of the geometrical properties of subspace must be designed to overcome the obstacle we are facing.

In this paper a novel detector for target signal buried in structured low-dimensional clutter was given. The main idea of our detector is that the geometrical characteristic of sampled data could be utilized in solving detection problem. Here the volume, a common concept for geometrical objects, was suitably defined for basis for subspaces (more concretely, the parallelotope with its edges being the basis vectors of this subspace). It is intuitive that the 'volume' of basis for low-dimensional subspaces in high-dimensional linear space is zero. So the judgment of whether or not the sampled vectors span a subspace that contains the target signal could be transformed naturally to the calculation of 'volume' of a parallelotope built by sampled vectors together with basis for target subspace. If the 'volume' is zero, then the conclusion can be drawn that the sample subspace contains the target signal subspace, otherwise the target vector must lie outside the sample subspace. Thus the volume-based subspace detector, instead of projection-based ones, can be used to cope with the problem of detecting target signal lying in known subspace under clutter background with unknown subspace structure.

Throughout this paper, we use small bold letters $\bm x$ to denote vectors, capital bold letters $\bm X$ to denote matrices(or subspaces); we use $\|\bm X\|_2$ and $\|\bm x\|_2$ to denote the $\ell_2$ norm of the matrix $\bm X$ and vector $\bm x$, and $\|\bm X\|_F$ to denote the Frobenius norm of matrix $\bm X$. The $d$-dimensional identity matrix is denoted by $\bm I_d$. $\Span(\bm X)$ represents the linear subspace spanned by column vectors of the matrix $\bm X$. $\bm{H}_1\oplus\bm{H}_2$ denotes the direct sum of subspaces $\bm{H}_1$ and $\bm{H}_2$. In addition, $\mathbb{P}$ and $\mathbb{E}$ denotes the probability and expectation respectively.

\par
The remainder of this paper is organized as follows: Some preliminary backgrounds on the geometrical concepts for linear subspaces such as principal angles and volumes were summarized in section II. Then the Volume-based Subspace Detector was introduced and demonstrated in detail in section III. In section IV, theoretical analysis on the performance of our volume-based subspace detector was given. The potential application and future work on volume-based subspace detector was discussed in final section.

\section{Preliminary Background}

In this section, some important concepts of linear space geometry were reviewed concisely. Although these results are fundamental for a deep understanding of linear subspace, these concepts rarely appear in common textbooks of linear algebra. Only necessary material for our discussion was put forward for the space limitation. For details, please see \cite{absil2004riemannian} and reference therein.

\subsection{Principal Angles between Subspaces}
\par
The concept of principal angles \cite{miao1992principal} is the natural generalization of that of angles between two vectors. Principal angles can be used to formulate the relationship between two subspaces.

\begin{Def}
For two linear subspaces $\bm{H}_1$ and $\bm{H}_2$, with dimensions $\dim(\bm{H}_1)=d_1,\dim(\bm{H}_2)=d_2$. Take $m = \min(d_1,d_2)$, then the principal angles $0 \leq \theta_1 \leq \cdots \leq \theta_m \leq \pi/2$ between $\bm{H}_1$ and $\bm{H}_2$ are defined by
\begin{eqnarray}
\cos \theta_i := \bm u_i^T \bm v_i = \max_{\bm u \in \bm{H}_1,\bm v \in \bm{H}_2}\left\{\bm u^T \bm v:\quad \begin{subarray}{c}\displaystyle{\|\bm u\|=1, \bm u^T \bm u_j=0}\\
							\displaystyle{\|\bm v\|=1,\bm v^T \bm v_j=0}
		\end{subarray},\quad j = 1,\cdots,i-1\right\} , \nonumber 
\end{eqnarray}
where $i =1, 2,\cdots m$.
\end{Def}

\par
As an important concept of linear space geometry, the principal angles are widely applied in scientific and engineering fields. For instance, the geodesic distance which is the key metric measure on Grassmann manifold, as well as numerous kinds of distance measures, is defined using the principal angles \cite{absil2004riemannian}\cite{qiu2005unitarily}\cite{edelman1998geometry}, such as

\begin{itemize}
\item Chordal Distance or Projection distance
$$
d_{proj}:= \left( \sum_{j=1}^{\min(d_1,d_2)}\sin^2 \theta_j(\bm{H}_1, \bm{H}_2) \right)^{1/2},
$$
\item Binet-Cauchy Distance
$$
d_{BC} := \left(1- \prod_{j=1}^{\min(d_1,d_2)} \cos^2 \theta_j(\bm{H}_1, \bm{H}_2)\right)^{1/2},
$$
\item Procrustes Distance
$$
d_{proc}:= 2 \left( \sum_{j=1}^{\min(d_1,d_2)}  \sin^2 \frac{\theta_j(\bm{H}_1, \bm{H}_2)}{2} \right)^{1/2}.
$$
\end{itemize}

Moreover, the volume of subspace which was used in this paper to construct our subspace detector is also closely related to the principal angles.

\subsection{The volume of a matrix}

\par
The definition of volume for certain geometrical object is ambiguous without the stipulation of dimension. For example, the volume of a parallelogram on a plane is the absolute value of cross product of its two adjacent sides. This is its two-dimensional volume. On the contrary, when the parallelogram is regarded as the three-dimensional body, its three-dimensional volume is definitely zero. It means that the volume value of an object depends on the dimension of space it lies in. For a square matrix
$\bm{X}\in\mathbb{R}^{n\times{n}}$, the $n$-dimensional volume of the parallelotope spanned by column vectors of $\bm{X}$ is well-known to be the absolute value of determinant of $\bm{X}$, which is the product of all the eigenvalues of $\bm{X}$. When $\bm{X}$ is rectangular, the concept of volume could be generalized naturally. Suppose a matrix $\bm{X}\in\mathbb{R}^{n\times{d}}$ with $d$ column vectors, and $d<n$, its $d$-dimensional volume is defined as \cite{ben1992volume}

\begin{equation}\label{VolumeDef1}
\Vol_d(\bm X):= \prod_{i=1}^d \sigma_i,
\end{equation}
where $\sigma_1 \geq \sigma_2 \geq \cdots \geq \sigma_d \geq 0$ are the singular values of $\bm{X}$. If $\bm X$ is of full column rank, its  $d$-dimensional volume can be written equivalently as \cite{ben1992volume}\cite{miao1992principal}
\begin{equation}\label{VolumeDef2}
\Vol_d (\bm X) = \sqrt{\det(\bm{X^TX})}.
\end{equation}

The following simple lemma is widely useful in application of volume for subspaces. It means that the $d$-dimensional volume of a matrix with a rank less than $d$ is definitely zero.

\begin{Lemma}\label{lemma2}
Suppose $\bm{X}^{(m)}=\{\bm{x}_1,\cdots,\bm{x}_m\}, m > 1$ be a group of vectors in Hilbert space $\bm H$ and $\dim(\Span{\bm{X}^{(m)}})=i$, then
\begin{equation}
\Vol_{d}(\bm{X}^{(m)})=0\quad \Longleftrightarrow \quad{i<d}.
\end{equation}
\end{Lemma}

The $d$-dimensional volume provides a kind of measure for separation between two linear subspaces. 
For the $n$-dimensional Hilbert space $\bm{H}$ and its two subspaces $\bm{H}_1$ and $\bm{H}_2$ with dimensions $\dim(\bm{H}_1)=d_1,\dim(\bm{H}_2)=d_2$, denote their bases matrices by $\bm X_1$ and $\bm X_2$, we define the \textit{volume-based correlation} as
\begin{equation}\label{CorrVol}
\Corr_{\textrm{vol}}(\bm{H}_1,\bm{H}_2)=\frac{\Vol_{d_1+d_2}([\bm{X}_1,\bm{X}_2])}{\Vol_{d_1}(\bm{X}_1)\Vol_{d_2}(\bm{X}_2)},
\end{equation}
where $[\bm{X}_1,\bm{X}_2]$ means putting columns of matrices $\bm X_1$ and $\bm X_2$ together.

The volume-based correlation is closely related to the principal angles between subspaces, i.e., according to \cite{miao1992principal},
\begin{equation}\label{VolAng}
\Corr_{\textrm{vol}}(\bm{H}_1,\bm{H}_2) = \prod_{j=1}^{\min(d_1,d_2)}\sin\theta_j(\bm{H}_1, \bm{H}_2),
\end{equation}
where $0\leq\theta_j(\bm{H}_1, \bm{H}_2)\leq2\pi,1\leq{j}\leq\min(d_1,d_2)$ are the principal angles of subspace $\bm{H}_1$ and $\bm{H}_2$.

\par
It can be seen intuitively from (\ref{VolAng}) that the volume-based correlation $\Corr_{\textrm{vol}}(\bm{H}_1,\bm{H}_2)$ can actually play the role of distance measure between subspaces $\bm{H}_1$ and $\bm{H}_2$. When $\bm{H}_1$ and $\bm{H}_2$ have vectors in common, i.e., $\Dim(\bm{H}_1\bigcap\bm{H}_2)\geq1$, we have $\Corr_{\textrm{vol}}(\bm{H}_1,\bm{H}_2)=0$. On the other side, when $\bm{H}_1$ is orthogonal to $\bm{H}_2$, we have $\Vol_{d_1+d_2}([\bm X_1, \bm X_2])=\Vol_{d_1}(\bm X_1)\Vol_{d_2}(\bm X_2)$, in other words, $\Corr_{\textrm{vol}}(\bm{H}_1,\bm{H}_2)=1$. Although volume-based correlation and the conventional correlation in statistics satisfy Cauchy-Schwarz inequality alike, they are essentially different, because the volume-based correlation isn't an inner product operation induced from some kind of distance in linear space of subspaces (more formally, Grassmann manifold). But it does not matter for the following discussion. The volume-based correlation will be seen as a generalized distance measure for convenience that plays a key role in our proposed subspace detector.

\section{The Volume-based Correlation Subspace Detector without random noise}

In order to fully convey the geometrical intuition about our subspace detector, in this section, we temporarily assume the noise component is not present, i,e., $\bm w=0$ in Problem 3. We introduce a novel detector called volume-based correlation subspace detector, or VC subspace detector for short, that can detect subspace signals buried in unknown, but usually high-power clutter. 
The main characteristic of the detector is that, it can exploit the geometric relation between the subspaces extracted from the sampled data and the target signal subspace, then it will eliminate the influence of clutter subspace gradually through the process of target detection.

\subsection{Main Idea}

As we have mentioned, the unknown clutter with an unknown subspace structure is the primary obstacle for efficient detection of target signal. To reach the purposes, the designers of detector must find a way to clarify the intrinsic construction of clutter subspace. Just as most of the traditional approaches for background learning, in our method, multiple samples are used to explore the clutter subspace. The following observation is the inspiration about the exploration of clutter subspace.

\begin{itemize}
  \item Suppose $\bm H$ be a $n$-dimensional Hilbert space, and $\bm{x}_1,\cdots,\bm{x}_k, k<n$ be randomly sampled vectors in $\bm H$, then in the generic situation, we have
      \begin{equation}
        \Dim(\Span\{\bm{x}_1,\cdots,\bm{x}_k\})=k,
      \end{equation}
      In other words, $\bm{x}_1,\cdots,\bm{x}_k$ are linearly independent.
  \item In the case of $k\geq{n}$, then in the generic situation, we have
      \begin{equation}
        \Dim(\Span\{\bm{x}_1,\cdots,\bm{x}_k\})=n,
      \end{equation}
      In a word, $\bm{x}_1,\cdots,\bm{x}_k$ are linearly dependent, but $\bm{x}_1,\cdots,\bm{x}_k$ span the whole space $\bm H$.
\end{itemize}

Let $\bm{H}_C$ be unknown clutter subspace with unknown dimension $d_1$, $\bm{H}_S$ be known target subspace with dimension $d_2$, we assume
$\Dim(\bm{H}_S\bigcap\bm{H}_C)=0$ throughout this paper. Denote $\bm{y}_1,\bm{y}_{2},\cdots,$ by the randomly sampled data satisfying the generic property mentioned above. The critical issue must be concerned with is that the sampled data might contain both clutter and target components in general, that is,
\begin{equation}
  \bm{y}_i = \bm{s}_i + \bm{c}_i, \quad{\bm{s}_i\in\bm{H}_S,\ \bm{c}_i\in\bm{H}_C},\ i=1,2,\cdots,
\end{equation}
Only with these $\bm y_i$, it is impossible to separate the clutter and target signal apart. According to the generic property of random sampling, we have
\begin{equation}
  \Dim(\Span\{\bm{y}_1,\cdots,\bm{y}_{d_1}\})=d_1.
\end{equation}
Even if the dimension $d_1$ was given virtually, the sample subspace $\Span\{\bm{y}_1,\cdots,\bm{y}_{d_1}\}$ still could not be regarded as the clutter subspace  $\bm{H}_C$, for the existence of target signal component $\bm s_i$. How could the sampled data be mined effectively to get knowledge of the clutter subspace?

We will show step by step that, the volume-based correlation between subspaces is helpful for us to eliminate the impact of mixing of clutter and target signal. 

Let $\bar{\bm s}_1,\cdots,\bar{\bm s}_{d_2}$ be the known basis vectors of $\bm{H}_S$. 
It has been mentioned that in the generic scenario of random sampling, different $\bm{y}_i$ sampled from $\bm H_S \oplus \bm H_C$ are linearly independent. In other words, innovative directions of basis vectors in $\bm{H}_S\oplus \bm H_C$ are revealed continually along with the sampling process. Combined with the known basis for signal subspace $\bm H_S$, we have
\begin{align}
  &\Dim(\Span\{\bm{y}_1,\bar{\bm s}_1,\cdots,\bar{\bm s}_{d_2}\})=1+d_2,\nonumber\\
  &\Dim(\Span\{\bm{y}_1,\bm{y}_2,\bar{\bm s}_1,\cdots,\bar{\bm s}_{d_2}\})=2+d_2,\nonumber\\
  &\cdots\cdots\nonumber\\
  &\Dim(\Span\{\bm{y}_1,\cdots,\bm{y}_{d_1},\bar{\bm s}_1,\cdots,\bar{\bm s}_{d_2}\})=d_1+d_2,
\end{align}
When we already have $d_1$ samples, actually $\bm{H}_C$ has been explored by successive sampling thoroughly, in other words, the projection of sample subspace $\Span\{\bm{y}_1,\cdots,\bm{y}_{d_1}\}$ on $\bm{H}_C$ has already constitute a complete basis for $\bm{H}_C$. The interesting point is, further sampling can not modify the intrinsic dimension of the subspace constituted by both the sampled data and basis for signal subspace, that is,
\begin{equation}
  \Dim(\Span\{\bm{y}_1,\cdots,\bm{y}_{k},\bar{\bm s}_1,\cdots,\bar{\bm s}_{d_2}\})=d_1+d_2,\quad{\text{for } k>d_1}
\end{equation}

The whole process about change of dimensions could be illustrated more clearly from the viewpoint of volume. 
In particular, the inspection of volume is the core idea of our proposed detector. 

Firstly, when both the signal and clutter are present, i.e., $\bm y_i \in \bm H_S\oplus \bm H_C$ we have
\begin{align}
&\Vol_{1+d_2}([\bm{y}_1,\bar{\bm s}_1,\cdots,\bar{\bm s}_{d_2}])>0,\nonumber\\
&\Vol_{2+d_2}([\bm{y}_1,\bm{y}_2,\bar{\bm s}_1,\cdots,\bar{\bm s}_{d_2}])>0,\nonumber\\
&\cdots\cdots\nonumber\\
&\Vol_{d_1+d_2}([\bm{y}_1,\cdots,\bm{y}_{d_1},\bar{\bm s}_1,\cdots,\bar{\bm s}_{d_2}])>0,
\end{align}
the magic will happen for the next dimension, i.e., when there are $d_1+d_2+1$ sample vectors, there will be
\begin{equation}
\Vol_{d_1+d_2+1}([\bm{y}_1,\cdots,\bm{y}_{d_1},\bm{y}_{d_1+1},\bar{\bm s}_1,\cdots,\bar{\bm s}_{d_2}])=0,\label{volscc}
\end{equation}
while on the other hand,
\begin{equation}
\Vol_{d_1+1}(\Span\{\bm{y}_1,\cdots,\bm{y}_{d_1},\bm{y}_{d_1+1}\})>0.\label{volsc}
\end{equation}
The reason is because, $d_1+d_2+1$ sample vectors in this scenario have not spanned the entire subspace $\bm H_S \oplus \bm H_C$ according to the previous statement of randomly sampling; but $\bm{y}_1,\cdots,\bm{y}_{d_1},\bm{y}_{d_1+1},\bar{\bm s}_1,\cdots,\bar{\bm s}_{d_2}$ can span $\bm H_S \oplus \bm H_C$, in another word,
\begin{eqnarray}
\Dim(\Span\{\bm{y}_1,\cdots,\bm{y}_{d_1},\bm{y}_{d_1+1},\bar{\bm s}_1,\cdots,\bar{\bm s}_{d_2}\})
=\Dim(\bm{H}_S\oplus\bm{H}_C)=d_1+d_2, \hspace{2cm}
\end{eqnarray}
On the other hand, if the sample data only contains pure clutter, i.e., $\bm y_i \in \bm H_C$, we obtain
\begin{equation}
\Vol_{d_1+d_2+1}([\bm{y}_1,\cdots,\bm{y}_{d_1},\bm{y}_{d_1+1},\bar{\bm s}_1,\cdots,\bar{\bm s}_{d_2}])=0,\label{volcc}
\end{equation}
and
\begin{equation}
\Vol_{d_1+1}([\bm{y}_1,\cdots,\bm{y}_{d_1},\bm{y}_{d_1+1}])=0.\label{volc}
\end{equation}
(\ref{volscc}), (\ref{volsc}), (\ref{volcc}) and (\ref{volc}) indicates that, $d_1+1$ is the critical number of samples for detection of target signal in the background of clutter with unknown subspace structure, i.e., the "breakpoint". The knack of detection in this noiseless situation is, sampling continually, computing the volume of parallelotope spanned by all the sample vectors and known basis of target subspace at various dimensions and inspect the change of results. Once the volume vanishes, it means the number of samples reaches the critical point. Then the process of sampling should be stopped and the volume of sample vector themselves is calculated. The decision can be made based on whether the result is zero, i.e., whether (\ref{volsc}) or (\ref{volc}).

\subsection{The Volume-based Correlation Subspace Detector on raw data}

Following the above thinking, an informal formulation of our detector in noiseless scenario, can be given below.

\textbf{Detector 0.}
\vspace{0.1cm}
\hrule
\vspace{0.1cm}

\begin{itemize}

\item Initial Step :

Obtain $\{\bar{\bm s}_1,\cdots,\bar{\bm s}_{d_2}\}$ as the basis vectors of signal subspace $\bm H_S$. Let the initial matrix of sample data be $\bm{Y}^{(0)}=\{0\}$. Index $i$ is set to $1$. Set two thresholds $T_1$ and $T_2$ at appropriate values.

\item Step 1 :

Get the new sample $\bm{y}_{i}$, let $\bm{Y}^{(i)}=[\bm{Y}^{(i-1)}, \bm{y}_{i}]$, then compute
\begin{equation}
  V^{(1)}(\bm{Y}^{(i)})=\Vol_{i+d_2}([\bm{Y}^{(i)},\bar{\bm s}_1,\cdots,\bar{\bm s}_{d_2}]),
\end{equation}

build the test quantity as
\begin{equation}\label{detector1}
  T^{(1)}(\bm{Y}^{(i)})=\frac{1}{V^{(1)}(\bm{Y}^{(i)})}, 
\end{equation}
if $T^{(1)}(\bm{Y}^{(i)})>T_1$, goto Step 2; else let $i=i+1$, goto Step 1,

\item Step 2 :

Compute
\begin{equation}
  V^{(2)}(\bm{Y}^{(i)})= 
  	\Vol_{i}([\bm{Y}^{(i-1)},\bm{y}_i]) 
\end{equation}

build the test quantity as
\begin{equation}
  T^{(2)}(\bm{Y}^{(i)})=\frac{1}{V^{(2)}(\bm{Y}^{(i)})},\label{idealtest2}
\end{equation}
if $T^{(2)}(\bm{Y}^{(i)})>T_1$, then concludes the non-existence of target signal; otherwise the conclusion is converse.

\end{itemize}
\hrule
\vspace{0.3cm}

It should be noted that the decision on existence of target can't be drawn by only examining the test quantity in (\ref{detector1}), because the volume of $\bm{Y}^{(i)}$ will also be zero when $\bm{Y}^{(i)}\subset\bm{H}_C$ and $i>\Dim(\bm{H}_C)$, leading naturally to zero value of $\Vol_{i+d_2}([\bm{Y}^{(i)},\bar{\bm s}_1,\cdots,\bar{\bm s}_{d_2}])$. That is to say, $T^{(1)}(\bm{Y}^{(i)})$ will become zero even when the target doesn't exist, because the dimension $i$ of matrix $\bm{Y}^{(i)}$ is larger than the dimension of clutter subspace. So additionally checking the volume of sample subspace $\bm{Y}^{(i)}$ is definitely necessary.

As we have said, when there is no target signal, we have 
\begin{align}
&V^{(1)}(\bm{Y}^{(i)})=0,\\
&V^{(2)}(\bm{Y}^{(i)})=0,
\end{align}
for $i\geq d_1 + 1$, but
on the other hand, when target does exist, we have
\begin{align}
&V^{(1)}(\bm{Y}^{(i)})=0,\\
&V^{(2)}(\bm{Y}^{(i)})>0,
\end{align}
if $d_1 + 1 \leq i\leq d_1 + d_2$, and $V^{(2)}(\bm{Y}^{(i)})$ will become zero only when $i \geq d_1 + d_2 +1$. So as a whole, when only clutter exists, both $V^{(1)}(\bm{Y}^{(i)})$ and $V^{(2)}(\bm{Y}^{(i)})$ will vanish simultaneously for the same $i$. While if target exists, they will become zero one after another as the increase of $i$. This is the essential sign of presence of target and the key point in our detector, and that's the reason why the detector relies on a joint test of (\ref{detector1}) and (\ref{idealtest2}).

\par

\subsection{The Volume-based Correlation Subspace Detector on orthogonalized data}

The fore-mentioned detector is impractical,
because there will be numerical stability problem 
when we are calculating volume of matrices with large dimensions in practical situation. Hence the procedure of Orthogonalization will be introduced into our detector. The advantages of orthogonalization include reducing the procedure of volume calculation and threshold testing from two steps to one, and improving the numerical stability dramatically.

Specifically, let $\bm{Q}_S$ be the matrix whose columns are the orthonormal basis for target signal subspace $\bm{H}_S$, which could be obtained off-line. The sample data $\bm{Y}^{(m)}=[\bm{y}_1,\bm y_2,\cdots,\bm y_m]$ taken from the subspace $\bm{H}_S\oplus\bm{H}_C$ (or $\bm H_C$) could be orthogonalized and normalized. We denote the result as matrix $\bm{Q}_{\bm{Y}}^{(m)}$. Then the test quantity (\ref{detector1}) in VC subspace detector, that is, the volume correlation between subspaces $\bm{H}_S$ and $\Span(\bm{Y}^{(m)})$ can be written as
\begin{equation}\label{testorth}
   T^{(1)}(\bm{Y}^{(m)})=\frac{1}{\Vol_{m+d_2}([\bm{Q}_{\bm{Y}}^{(m)},\ \bm{Q}_S])},
\end{equation}
because of the fact that
\begin{equation}
  \Vol_{m}(\bm{Q}_{\bm{Y}}^{(m)})=1,\qquad\Vol_{d_2}(\bm{Q}_S)=1,
\end{equation}
we have
\begin{equation}
   T^{(1)}(\bm{Y}^{(m)})=\frac{1}{\Corr_{\textrm{vol}}(\Span(\bm{Y}^{(m)}),\bm{H}_S)}.
\end{equation}
The matrix $\bm{Q}_S$ in (\ref{testorth}) can be prepared in advance. and matrix $\bm{Q}_{\bm{Y}}^{(m)}$ can be generated recursively as
\begin{equation}
  \bm{Q}_{\bm{Y}}^{(m)}=\left[\bm{Q}_{\bm{Y}}^{(m-1)},\
  \frac{(\bm{I}-\bm{Q}_{\bm{Y}}^{(m-1)}(\bm{Q}_{\bm{Y}}^{(m-1)})^T)\bm{y}_m}{\|(\bm{I}-\bm{Q}_{\bm{Y}}^{(m-1)}(\bm{Q}_{\bm{Y}}^{(m-1)})^T)\bm{y}_m\|}\right].
\end{equation}

It should be noted that the two-step test in our previous VC subspace detector could be reduced to just one with the help of orthogonalization. In fact, if there is no target signal in received data, we have
\begin{equation}
  \Vol_1(\bm{Q}_{\bm{Y}}^{(1)})>0,\ \Vol_2(\bm{Q}_{\bm{Y}}^{(2)})>0,\ \cdots,\Vol_{d_1}(\bm{Q}_{\bm{Y}}^{(d_1)})>0.
\end{equation}
Contrast to the case without orthogonalization, the volume of $\bm{Q}_{\bm{Y}}^{(m)}$ will not become zero when $m>d_1$. Because until now, the innovative vector $\bm{y}_m$ could be linearly expressed by columns of $\bm{Q}_{\bm{Y}}^{(m)}$. So if
$$
(\bm{I}-\bm{Q}_{\bm{Y}}^{(m-1)}(\bm{Q}_{\bm{Y}}^{(m-1)})^T)\bm{y}_m=\bm 0,
$$
we have
\begin{equation}
\bm{Q}_{\bm{Y}}^{(m+1)}=\bm{Q}_{\bm{Y}}^{(m)} \text{ and } \Rank(\bm{Q}_{\bm{Y}}^{(m)}) = d_1,\qquad{m\geq{d_1}},
\end{equation}
hence
\begin{equation}
\Vol_{d_1}(\bm{Q}_{\bm{Y}}^{(m+1)})=\Vol_{d_1}(\bm{Q}_{\bm{Y}}^{(m)})=1,\qquad{\text{for }m\geq{d_1}}.
\end{equation}

Similarly, when there exists target signal, we also have 
\begin{equation}
\Vol_{d_1 + d_2}(\bm{Q}_{\bm{Y}}^{(m+1)})=\Vol_{d_1 + d_2}(\bm{Q}_{\bm{Y}}^{(m)})=1,\qquad{\text{for }m\geq{d_1+d_2}}.
\end{equation}
As we can see, the volume of $\bm{Q}_{\bm{Y}}^{(m)}$ will never become zero.

When we are considering the other test quantity, according to (\ref{VolAng}), when there is no target signal, we have
\begin{equation}
\Vol_{d_1+d_2}([\bm{Q}_{\bm{Y}}^{(m)},\ \bm{Q}_S])=\prod_{j=1}^{\min(d_1,d_2)}\sin\theta_j(\bm{H}_C, \bm{H}_S)>0, \text{ for } m \geq d_1,
\end{equation}
because $\dim(\bm H_S \bigcap \bm H_C)=0$. 

On the other hand, when there exists target signal,
according to the analysis in subsection A, it is obvious that 
\begin{equation}
  \Vol_{d_1+d_2}([\bm{Q}_{\bm{Y}}^{(m)},\ \bm{Q}_S])=0,\text{ for } m \geq d_1, 
\end{equation}
because $\Span(\bm{Y}^{(m)})\cap\bm{H}_S\neq\{0\}$.
As a result, the detector only needs a test on (\ref{testorth}), because it is not possible that the volume of $\bm{Q}_{\bm{Y}}^{(m)}$ equals zero whether or not the target presents in received data, 
orthogonalization eliminates completely the possibility of rank deficiency of matrix $\bm{Q}_{\bm{Y}}^{(m)}$. There is no need to check the value of $\Vol_m(\bm{Y}^{(m)})$ alone in VC subspace detector at all.

The detector could be adapted as follows,

\textbf{Detector 1.}
\vspace{0.1cm}
\hrule
\vspace{0.1cm}

\begin{itemize}

\item Initial Step :

Obtain $\{\bar{\bm s}_1,\cdots,\bar{\bm s}_{d_2}\}$ as the orthonormal basis vectors of target subspace and denote it by $\bm{Q}_S$. Let the initial matrix of sample data be $\bm{Y}^{(0)}=\{0\}$. Index $i$ is set to $1$. Set two threshold $T$ and $\epsilon$ at appropriate values.

\item Step 1 :

Get the new sample $\bm{y}_{i}$, let
\begin{equation*}
  \bm{Q}_{\bm{Y}}^{(i)}=\left[\bm{Q}_{\bm{Y}}^{(i-1)},\
  \frac{(\bm{I}-\bm{Q}_{\bm{Y}}^{(i-1)}(\bm{Q}_{\bm{Y}}^{(i-1)})^T)\bm{y}_i}{\|(\bm{I}-\bm{Q}_{\bm{Y}}^{(i-1)}(\bm{Q}_{\bm{Y}}^{(i-1)})^T)\bm{y}_i\|}\right],
\end{equation*}
for $\bm{y}_i\notin\Span(\bm{Q}_{\bm{Y}}^{(i-1)})$; while we let $\bm{Q}_{\bm{Y}}^{(i)}=\bm{Q}_{\bm{Y}}^{(i-1)}$ otherwise.

Then we compute the test quantity as
\begin{equation}\label{detector2}
  T(\bm{Y}^{(i)})=\frac{1}{\Vol_{i+d_2}([\bm{Q}_{\bm{Y}}^{(i)}, \bm{Q}_S])},
\end{equation}
if $T(\bm{Y}^{(i)})>T$, concludes the existence of target signal and exits;

\item Step 2 :

if $|T(\bm{Y}^{(i)})-T(\bm{Y}^{(i-1)})|>\epsilon$, then set $i=i+1$ and go to step 1; otherwise concludes the non-existence of target signal and exits.

\end{itemize}
\hrule
\vspace{0.3cm}
There are some further remarks deserve being mentioned explicitly.

\begin{itemize}
	
	\item Remark 1. The core of the proposed noiseless VC subspace detector is a volume-based test, the final decision is made based on the test of (\ref{detector2}), which is actually the reciprocal of the volume correlation between the target signal subspace and the sample subspace. According to the previous analysis, when target exists in sampled data, $T(\bm{Y}^{(i)})$ will reach infinity for $i \geq d_1$, while on the other hand, $T(\bm{Y}^{(i)})$ will remains a finite value if only clutter exists. The breakpoint of dimension for volume computation could be discovered AUTOMATICALLY, which is the indication of unknown dimension of clutter subspace. The reasons can be described briefly as that results of volume-based correlations is independent of the intrinsic structure of subspaces, and depends only on the dimensions and mutual relationship of the subspaces. So VC subspace detector focuses on the evolution of values of volume-based correlations along with increase of dimensions for volume computation only, regardless of the basis structure of clutter subspace. Concerning with the unknown characteristic of clutter subspace, VC subspace detector could be listed among the blind detecting methods. It is important to note that although matched subspace detector also detect signals from clutters, these two detection methods are essentially different. Because the prerequisite for matched subspace detector includes the detailed information on clutter subspace. However, it is needless for VC subspace detector.
	
	\item Remark 2. It should be emphasized that the most remarkable advantage of VC subspace detector is its feature of "Detecting while Learning". To be specific, the detection could be completed without separated sessions for background learning with VC subspace detector. As is well known, background learning is very popular in adaptive processing for radar, communication and other signal processing problems. Channel equalization in communication transmission, CFAR (Constant False Alarm Rate) operation in radar detection and estimation of covariance matrices for clutter echoes in STAP (Space-Time Adaptive Processing) all belong to sessions of background learning. There are double common defects for all these schemes. The first is that the efficacy of estimating clutter background might be influenced heavily by existence of target signal, which is referred to as target leakage in literatures; the second is the non-homogeneousness widely existed in clutter environment which easily leads to mismatch of the consequence of learning with the actual clutter scenario at the target location. Nevertheless, VC subspace detector stands far away from these trouble because the process of background process is accomplished implicitly and simultaneously with the detection operation. While the raw data are being sampled and put into the detector sequentially, the volume-based correlations are examined and tested constantly until the breakpoint is reached. The information of clutter subspace is being learned in the form of volumes of basis for sample subspace. At the breakpoint, background learning ends while simultaneously the decision on the existence of target is made naturally. There is no need for extra effort of background learning. The learning and detection procedures are merged perfectly in VC subspace detector. We call this interesting property "Detecting while Learning".
	
	\item Remark 3. The acceptance criterion of VC subspace detector for the hypothesis on existence of target signal is whether or not certain volume-based correlations are zero. It is generally believed that testing some quantities to be zero is impractical. Lots of statistical methods, such as MUSIC and other subspace-class algorithms use reciprocal to transform the value near zero to be relatively large. Hence the difference between negligible results could be sharpen and the power of detection is greatly strengthened. VC subspace detector is not an exception.
	
	\item Remark 4. In noisy situation, the detection problem becomes more complicated. Random noise will disturb our judgement and must be eliminated to assure the quality of detector. There are plenty of mature techniques for extracting the informative subspaces from noise, such as eigen-decomposition based filtering and subspace tracking, for us to choose as the preprocessing steps of VC subspace detector. Some detailed discussion on these issues will be given in section IV.
	
\end{itemize}

\subsection{Theoretical property of the proposed detector in noiseless situation}

\par
We will show theoretically that our volume-based correlation subspace detector can totally eliminate the influence of the clutter and detect the target signal in known target subspace effectively. Besides, an interesting monotonicity property for volumes-based correlation values of subspaces with different dimension will also be given, which provide a theoretical insight of our subspace detector. 

\begin{Theorem}\label{H1Thm1}
Let $\bm H$ be the $n$-dimensional Hilbert space, $\bm{H}_S$ and $\bm{H}_C$ be the subspace of $\bm H$ corresponding to target and clutter respectively. $\bm{H}_S\cap\bm{H}_C=\{0\}$, $\Dim{\bm{H}_C}=d_1$, $\Dim{\bm{H}_S}=d_2$, Suppose $\bm{y}_i, i=1,2,\cdots$ be randomly sampled data either containing both target and clutter,
\begin{equation}
\bm{y}_i = \bm s_i+\bm c_i,\qquad i = 1,2,\cdots
\end{equation}
or only containing clutter
\begin{equation}
\bm{y}_i = \bm c_i,\qquad i = 1,2,\cdots
\end{equation}	
where $\bm s_i\in\bm{H}_S$ and $\bm c_i\in\bm{H}_C$. Let
\begin{equation}
\bm{Y}^{(m)}:=[\bm{y}^{(1)},\cdots,\bm{y}^{(m)}],
\end{equation}
and $\bm{Q}_{\bm{Y}}^{(m)}$ and $\bm{Q}_S$ be orthogonal matrix with columns being the basis vectors of $\Span(\bm{Y}^{(m)})$ and $\bm{H}_S$, then we have the following monotone property
\begin{equation}\label{Monotone}
V(\bm{Y}^{(1)})\geq V(\bm{Y}^{(2)})\geq\cdots\geq V(\bm{Y}^{(d_1)})
\end{equation}
where
\begin{equation}
V(\bm{Y}^{(m)})=\Vol_{m+d_2}([\bm{Q}_{\bm{Y}}^{(m)},\bm{Q}_S]),\quad m=1,2,\cdots,d_1.
\end{equation}
\end{Theorem}

The monotone property formulated in Theorem \ref{H1Thm1} might be useful when VC subspace detector was used in practical scenario. It guarantees that the test quantity of our detector will increase continuously with the dimension of sample subspace until the breakpoint is reached. So the breakpoint should be found easily without the annoying fluctuation of test quantity.

\begin{Theorem}\label{H0Thm2}
Under the same assumption of theorem \ref{H1Thm1}, the sufficient and necessary condition for existence of target signal in sample subspace is there is an integer $K<n$ such that
\begin{equation}
V(\bm{Y}^{(K)})=\Vol_{K+d_2}([\bm{Q}_{\bm{Y}}^{(K)},\bm{Q}_S])=0.\label{corr1}
\end{equation}
\end{Theorem}

Theorem \ref{H1Thm1} and \ref{H0Thm2} theoretically guarantee the behavior of our volume-based detector, as is shown in figure \ref{figure1}, the volume quantity $V(\bm{Y}^{(m)})$ will gradually drop as the increase of samples, it will drop to zero at index $K$ if and only if target signal exists; and remain non-zero when only clutter exists.
The index $K$ in Theorem \ref{H0Thm2} is exactly the dimension of clutter subspace $\bm{H}_C$ when the target signal is presented in sample subspace. Concerning with the monotone property in Theorem \ref{H1Thm1}, $K$ is the smallest index for sample subspace $\Span(\bm{Y}^{(m)})$ to satisfy (\ref{corr1}). Such $K$ could be called the critical point or phase transition point, for it indicates the essential change of volume correlation between sample subspace and target signal subspace. The reason for this change must be rank deficiency of direct sum of sample subspace and target subspace, which actually indicates the existence of target components in sample subspace.
The behavior of our VC detector described in theorem \ref{H1Thm1} and \ref{H0Thm2} is illustrated in figure \ref{figure1}, in which we randomly generate two subspaces as $\bm H_S$ and $\bm H_C$, and plot the variation of volume-based correlation $V(\bm Y^{(m)})$ with increase of $m$ in Detector 1.

\begin{figure}[htbp]
\includegraphics[width=0.9\textwidth]{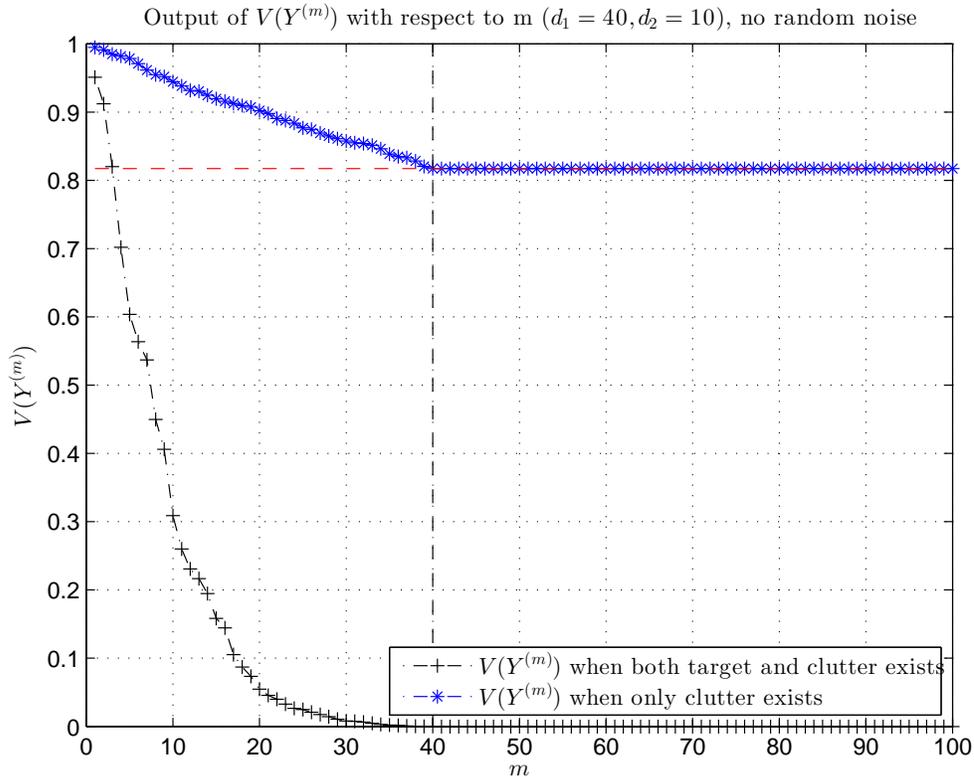}
\caption{Simulation of the volume-correlation $V(\bm Y^{(m)})$ in noiseless situation}
		\label{figure1}
\end{figure}

\section{The Volume-Correlation Subspace Detector in noisy environment}
\par
In this section, we will modify our VC subspace detector to make it be more suitable for noisy environment. Then some asymptotic result on the performance of VC subspace detector will be given.

\subsection{Main Idea}

\par
The main problem here is the sample subspace has been contaminated by random noise and can not be used directly to compute the volume correlation in VC subspace detector. Therefore the noise must be cleared in advance. For most statistical signal processing algorithms concerned with subspaces, such as MUSIC, ESPRIT and so on, the target signal and random noise are separated into signal subspaces and noise subspaces by eigen-decomposition of correlation matrices firstly for further treatment. It implies the natural strategy of denoising for subspace-based signal processing. That is extracting signal subspaces for follow-up analysis and discarding noise subspaces simply for the purpose of noise elimination. In our case, the theme of VC subspace detector is detecting target signal from known target subspace lying in unknown clutter subspace. The main obstacle for our detector is how to accomplish subspace detection under the background of deterministic but sealed clutter. Hence random noise isn't the critical factor in the detection and will be treated concisely.

To be specific, let $\bm H$ be $n$-dimensional Hilbert space, $\bm{H}_S$ and $\bm{H}_C$ be target and clutter subspaces of $\bm H$ respectively, and $\bm{H}_S\cap\bm{H}_C=\{0\}$, $\Dim{\bm{H}_C}=d_1$, $\Dim{\bm{H}_S}=d_2$, the signal model for our detector is
$$
\bm{r}=\bm{y}+\bm{w},
$$
where 
$\bm{w}$ is the noise vector usually assumed to be white and Gaussian distributed with mean 0 and variance $\sigma^2$, and $\bm{y}$ represents the signal components
being sampled randomly from $\bm{H}_S\oplus\bm{H}_C$ or just $\bm H_C$. The deterministic vector $\bm{y}$ can be randomized and written as
$$
\bm{y}=\bm{X}\bm{\alpha},
$$
where 
$\Span(\bm{X})=\bm{H}_S\oplus\bm{H}_C$ or $\bm H_C$, and $\bm{\alpha}$ is a random vector with finite second order moments. Then the correlation matrix $\mathrm{R}_r$ of the sample data $\bm{r}$ is
\begin{equation}\label{CorrMatrix}
\mathrm{R}_{\bm r}=\mathbb{E}\{\bm{r}\bm{r}^T\}=\bm{X}\mathbb{E}\{\bm{\alpha}\bm{\alpha}^T\}\bm{X}^T+\sigma^2\bm I_n
\end{equation}
without loss of generality, $\bm{X}$ is assumed to be full rank, we denote the column rank of $\bm X$ by $k$, then the eigenvalues of $\mathrm{R}_r$ could be listed as
\begin{equation}\label{realEigenvalue}
\lambda_1\geq\lambda_2\geq\cdots\geq\lambda_k\geq\lambda_{k+1}=\cdots=\lambda_n=\sigma^2,
\end{equation}
and the corresponding eigenvectors are
$$
\bm{q}_1, \bm{q}_2,\cdots,\bm{q}_k,\bm{q}_{k+1},\cdots,\bm{q}_n
$$
Denote $\bm{Q}_{SC}:=[\bm{q}_1, \bm{q}_2,\cdots, \bm{q}_k]\in\mathbb{R}^{n\times{k}}$, $\bm{Q}_{N}:=[\bm{q}_{k+1}, \bm{q}_{k+2},\cdots, \bm{q}_n]\in\mathbb{R}^{n\times{n-k}}$.  It is clear that
\begin{equation}
\Span(\bm{Q}_{SC}) = \bm H_S \oplus \bm H_C,\quad k=d_1+d_2,
\end{equation}
when both signal and clutter are present, and
\begin{equation}
\Span(\bm{Q}_{SC}) =  \bm H_C,\quad k=d_1,
\end{equation}
when the sampled data contains "pure" clutter.
$\Span(\bm{Q}_{SC})$ and $\Span(\bm{Q}_{N})$ are commonly called signal subspace and noise subspace. Therefore $\Span(\bm{Q}_{SC})$ could be used as proxy of $\bm{H}_S\oplus\bm{H}_C$ (or $\bm H_C$) and the main idea in previous section is workable as well in the noisy environment.

\subsection{VC Subspace Detector in noisy environment}

The VC subspace detector could be extended to noisy scenario as follows:

\textbf{Detector 2.}
\vspace{0.1cm}
\hrule
\vspace{0.1cm}
\begin{itemize}

\item Initial Step :

Denote the received data $\{\bm{r}_1,\cdots,\bm{r}_m\}$ by $\bm{R}^{(m)}$. Obtain $\{{\bm s}_1,\cdots,{\bm s}_{d_2}\}$ as the orthonormal basis vectors of target subspace and denote it by $\bm{Q}_S$. Let the sample covariance matrix be $\hat{\mathrm{R}}^{(0)}=0$. Index $i$ is set to $1$. Set two thresholds $T$ and $\epsilon$ at appropriate values.

\item Step 1 :

Get the new sample $\bm{r}_{i}$, compute the covariance matrix as
\begin{equation}\label{CorrMat}
  \hat{\mathrm{R}}^{(i)}=\frac{i-1}{i}\hat{\mathrm{R}}^{(i-1)}+\frac{1}{i}\bm{r}_{i}\bm{r}_{i}^{\rm T},
\end{equation}
Assume the eigenvalues of $\hat{\mathrm{R}}^{(i)}$ be
\begin{equation}\label{eigenvalue}
\hat{\lambda}_1\geq\hat{\lambda}_2\geq\cdots\geq\hat{\lambda}_{k_i}\geq\hat{\lambda}_{k_i+1}\geq\cdots\geq\hat{\lambda}_n,
\end{equation}
and the corresponding eigenvectors be
\begin{equation}\label{eigenvector}
\hat{\bm{q}}_1, \hat{\bm{q}}_2,\cdots,\hat{\bm{q}}_{k_i},\hat{\bm{q}}_{k_i+1},\cdots,\hat{\bm{q}}_n,
\end{equation}
determine the dimension of signal subspace $k_i$, and
construct the estimated sample subspace $\hat{\bm{U}}^{(i)}$ as
\begin{equation}
\hat{\bm{Q}}^{(i)}=[\hat{\bm{q}}_1, \hat{\bm{q}}_2,\cdots,\hat{\bm{q}}_{k_i}]
\end{equation}

\item Step 2 :

Compute the test quantity as
\begin{equation}
  T(\bm{R}^{(i)})=\frac{1}{\Vol_{k_i+d_2}([\hat{\bm{Q}}^{(i)},\ \bm{Q}_S])},
\end{equation}
if $T(\bm{R}^{(i)})>T$, concludes the existence of target signal and exits;

\item Step 3 :

if $|T(\bm{R}^{(i)})-T(\bm{R}^{(i-1)})|>\epsilon$, then set $i=i+1$ and go to step 1; otherwise concludes the non-existence of target signal and exits.

\end{itemize}
\hrule
\vspace{0.3cm}

As described above, the subspace $\hat{\bm{Q}}^{(i)}$ built from eigenvectors corresponding to large eigenvalues of covariance matrix of sampled data was taken to be the sample subspace in VC subspace detector. It is because the true covariance matrices can't be obtained straight from sample data such that the sample covariance matrices were calculated via (\ref{CorrMat}) instead.
The dimension of signal-plus-clutter subspace (or clutter subspace), i.e., $k_i$ in (\ref{eigenvalue}) actually needs to be estimated. There are various methods can be used, from the conventional AIC or MDL and their variations \cite{akaike1974new,wax1985detection}, to the newest Bayesian Information Criterion (BIC \cite{schwarz1978estimating}, GBIC \cite{lu2013generalized}), Random Matrix Theory (RMT, \cite{kritchman2009non}) and Entropy Estimation of Eigenvalues (EEE, \cite{asadi2013source}), etc. Since they all belong to another domain of research, we will not discuss this topic in detail here. In the following analysis, we just assume the dimension $k_i$, which is related with $d_1$ and $d_2$, is accurately known or estimated. With the eigen-decomposition method,
 the accuracy about this approximation of subspace $\bm H_S\oplus \bm H_C$ (or $\bm H_C$) had been studied extensively \cite{Stoica1991Statistical}\cite{stoica1989music}\cite{jeffries1985asymptotic} and the feasibility of $\hat{\bm{Q}}^{(i)}$ had been proved asymptotically. Hence we use it in VC subspace detector as the substitution of sample subspace when noise is presented. We can expect the proposed VC subspace detector will asymptotically approximate the VC subspace detector in noiseless scenario, and this expectation is validated by the theory in next section.

\subsection{Property of Detector}

To avoid the vagueness brought by asymptotic conclusion of the performance of VC subspace detector in the noisy background, we give some non-asymptotic analysis on the capability of our detector with knowledge of random matrices and concentration inequalities. Denote the received data by $\{\bm{r}_1,\cdots,\bm{r}_m\}$, and denote the correlation matrix by $\mathrm{R}_r$ as in (\ref{CorrMatrix}). 
Suppose the eigenvalues of $\mathrm{R}_r$ be as (\ref{eigenvalue}) and its eigenvectors be as (\ref{eigenvector}), We have

\begin{Theorem}\label{H1Thm1N}
Let $H$ be $n$-dimensional Hilbert space, $\bm{H}_S$ and $\bm{H}_C$ be target and clutter subspaces of $\bm H$ respectively, $\dim(\bm{H}_S\cap\bm{H}_C)=0$, $\Dim{\bm{H}_C}=d_1$, $\Dim{\bm{H}_S}=d_2$,
\begin{equation}
\bm{r}_i=\bm{y}_i+\bm{w}_i,\qquad i=1,2,\cdots,m.\nonumber
\end{equation}
where $\bm{r}_i$ is the sample data, $\bm{y}_i\in\bm{H}_S\oplus\bm{H}_C$, $\bm{w}_i\thicksim\mathcal{N}(0,\sigma^2\bm{I}_n)$ is Gaussian white noise.

If the target signal presents in sample data, then for any $0<\varepsilon<1$ and $\delta>0$, if
\begin{equation}\label{Res8}
m\geq\frac{1+\varepsilon}{(\sqrt{\delta+1}-1)^2}\left(\sum_{\stackrel{i,j=1}{j\neq i}}^{d_2+d_1} \frac{\lambda_i\lambda_j}{(\lambda_i-\lambda_j)^2}+(n-d_1-d_2)\sum_{i=1}^{d_2+d_1}\frac{\lambda_i\sigma^2}{(\sigma^2-\lambda_i)^2}\right),
\end{equation}
then there exists a constant $C>0$, such that
\begin{equation}\label{Res7}
\frac{1}{T^2(\bm{R}^{(m)})}\leq\delta^{d_2} + O(\delta^{d_2+1})
\end{equation}
holds with probability
\begin{equation}\label{prob1}
\mathbb{P}\geq1-\exp\{-\frac{(d_1+d_2)\cdot{n}\cdot\varepsilon^2}{C}\}.
\end{equation}

On the contrary, in the case of non-target, for any $0<\varepsilon<1$ and $\delta>0$, when
\begin{eqnarray}\label{Res9}
	m \geq\frac{1+\varepsilon}{(\sqrt{\delta+1}-1)^2}\cdot\left(  \sum_{i =1}^{d_1} \sum_{\stackrel{j=1}{j\neq i}}^{d_1} \frac{\lambda_i \lambda_j}{(\lambda_i-\lambda_j)^2}+\sum_{i =1}^{d_1}(n-d_1)\frac{\lambda_i\sigma^2}{(\sigma^2-\lambda_i)^2}  \right),
\end{eqnarray}
we have
\begin{equation}\label{Res10}
	|\frac{1}{T^2(\bm R^{(m)})}-\tau^2(\bm{H}_S,\bm{H}_C)|\leq s_{d_1-1}(\bm{Q}_C^T\bm{P}_S^{\perp}\bm{Q}_C)\delta + O(\delta^{2}),
\end{equation}
holds with probability
\begin{equation}\label{prob2}
	\mathbb{P} \geq 1-\exp\{-\frac{d_1\cdot{n}\cdot \varepsilon^2}{C}\},
\end{equation}
here $\tau(\bm{H}_S,\bm{H}_C)>0$ is a constant related with $\bm{H}_S$ and $\bm{H}_C$, $\bm Q_S$, $\bm{Q}_C$ are the orthogonal bases of $\bm H_S$ and $\bm H_C$, respectively, and $\bm{P}_S^{\perp}$ is the projection matrix onto the orthogonal complement of $\bm{H}_S$, $s_{k}(\bm A)$ for $n \times n$ square matrix $\bm A$ is defined as:
\begin{equation}\label{ESF}
s_{k}(\bm A) := \sum_{1 \leq i_1 \leq \cdots \leq i_k \leq n} \sigma_{i_1}\cdots \sigma_{i_k}, 1\leq k \leq n,
\end{equation}	
where $\sigma_1,\cdots,\sigma_n$ are singular values of matrix $\bm A$.
\end{Theorem}

Theorem \ref{H1Thm1N} describes the performance of our VC subspace detector in noisy environment. The main result (\ref{Res7}), together with (\ref{Res8}) and (\ref{prob1}), implies that when the target signal is present, the test quantity $T(\bm R^{(m)})$ of VC subspace detector will tend to infinity with an overwhelming probability, when the number of sample data is sufficient large. On the other hand, (\ref{Res10}) together with (\ref{Res9}) and (\ref{prob2}) ensures $T(\bm R^{(m)})$ to tend to a finite value. Therefore the decision point of this detector is to test whether $T(\bm R^{(m)})$ increases over a threshold, or stops increasing at a finite value. The result of Theorem \ref{H1Thm1N} implies that the output of our VC subspace detector will yield two different results in the two mentioned scenarios, no matter what clutter is given, whether or not there is noise, so far as that we have enough sample data. This shows the asymptotic effectiveness of our VC subspace detector.

The results (\ref{Res7}) and (\ref{Res10}) of theorem \ref{H1Thm1N} depends on the deviation parameter $\delta$. It determines the order for infinitesimal deviation between the computational result of volume correlation and its ideal value when there is no noise. The coefficient of leading term in (\ref{Res10}) is the value of  $s_{d_1-1}(\bm{Q}_{C}^T\bm P_{S}^{\perp}\bm Q_{C})$, the elementary symmetry functions of singular values of $\bm{Q}_{C}^T\bm P_{S}^{\perp}\bm Q_{C}$. It only depends on the mutual relation between subspaces $\bm{H}_C$ and $\bm{H}_S$.

It should be mentioned that the lower bound of observation size, $m$ in (\ref{Res8}) and (\ref{Res9}), is a sufficient condition for the conclusion of theorem. It indicates that for a given parameter $\delta$, when $m$ satisfies (\ref{Res8}) and (\ref{Res9}), then the volume-correlation will sufficiently satisfy (\ref{Res7}) and (\ref{Res10}) with the high probability. It is known that sufficient conditions are usually more conservative and strict than that needed. Therefore, the value of volume-correlation would converge much more faster in practise than what is depicted in the condition (\ref{Res8}) and (\ref{Res9}). This will also be illustrated in the later numerical simulations.

The effectiveness of detector 1 was demonstrated by numerical simulation in figure \ref{figure3}. Here $n=1024$, $d_1=40,d_2=10$ and the target and clutter signal were chosen randomly from corresponding subspaces. In the simulation, we assume the dimensions of signal subspace, i.e., $k_i$ in (\ref{eigenvalue}) to be accurately known beforehand, and the Signal-to-Noise Ratio is 0dB, then the mean and individual values of the volume-correlation $T(\bm{R}^{(m)})$ with respect to different $m$ in 1000 monte-carlo simulations are plotted. Here the result for each simulation was illustrated by the scatter diagram in the small sub-figures. As $m$ increases, $T(\bm{R}^{(m)})$ converges to infinity when target signal exists, while in case of no target signal, $T(\bm{R}^{(m)})$ converges to a finite value. Therefore, as a whole, the simulation result verified the validity of VC subspace detector.

It is mentioned that the eigen-decomposition step in detector 2 is just a noise reduction procedure, aiming at getting the estimated basis $\hat{\bm Q}^{(m)}$ for subspace $\bm H_S \oplus \bm H_C$ (or $\bm H_C$) from the sampled data. It is obvious that in this setting, when the noise's power approaches zero, i.e., when the $\sigma$ in (\ref{realEigenvalue}) tends to 0, the signal subspace $\Span(\bm{Q}^{(m)})$ will be exactly the same as $\bm H_S \oplus \bm H_C$ (or $\bm H_C$), and the eigen-decomposition procedure in Step 1 of Detector 2 is equivalent to the orthogonalization procedure in Detector 1. Thus, it is expected that the performance of Detector 2 will be the same as that of Detector 1 when noise vanishes.
The comparison about the output of Detector 2 in noisy environment and that of Detector 1 is illustrated in figure \ref{figure4}, is is obvious that, when there is noise, the volume-based correlation $V(\hat{\bm{Q}}^{(m)}):=\Vol_{k_m+d_2}([\hat{\bm{Q}}^{(m)},\ \bm{Q}_S])$ in Detector 2 will slightly deviate from the $V(\bm{Q}_{\bm{Y}}^{(m)})=\Vol_{m+d_2}([\bm{Q}_{\bm{Y}}^{(m)}, \bm{Q}_S])$ in Detector 1; but as the noise's power drops, the value of $V(\hat{\bm{Q}}^{(m)})$ in Detector 2 converges to $V(\bm{Q}_{\bm{Y}}^{(m)})$ in Detector 1.

\begin{figure}[htbp]
	\centering
	\includegraphics*[width=0.9\textwidth]{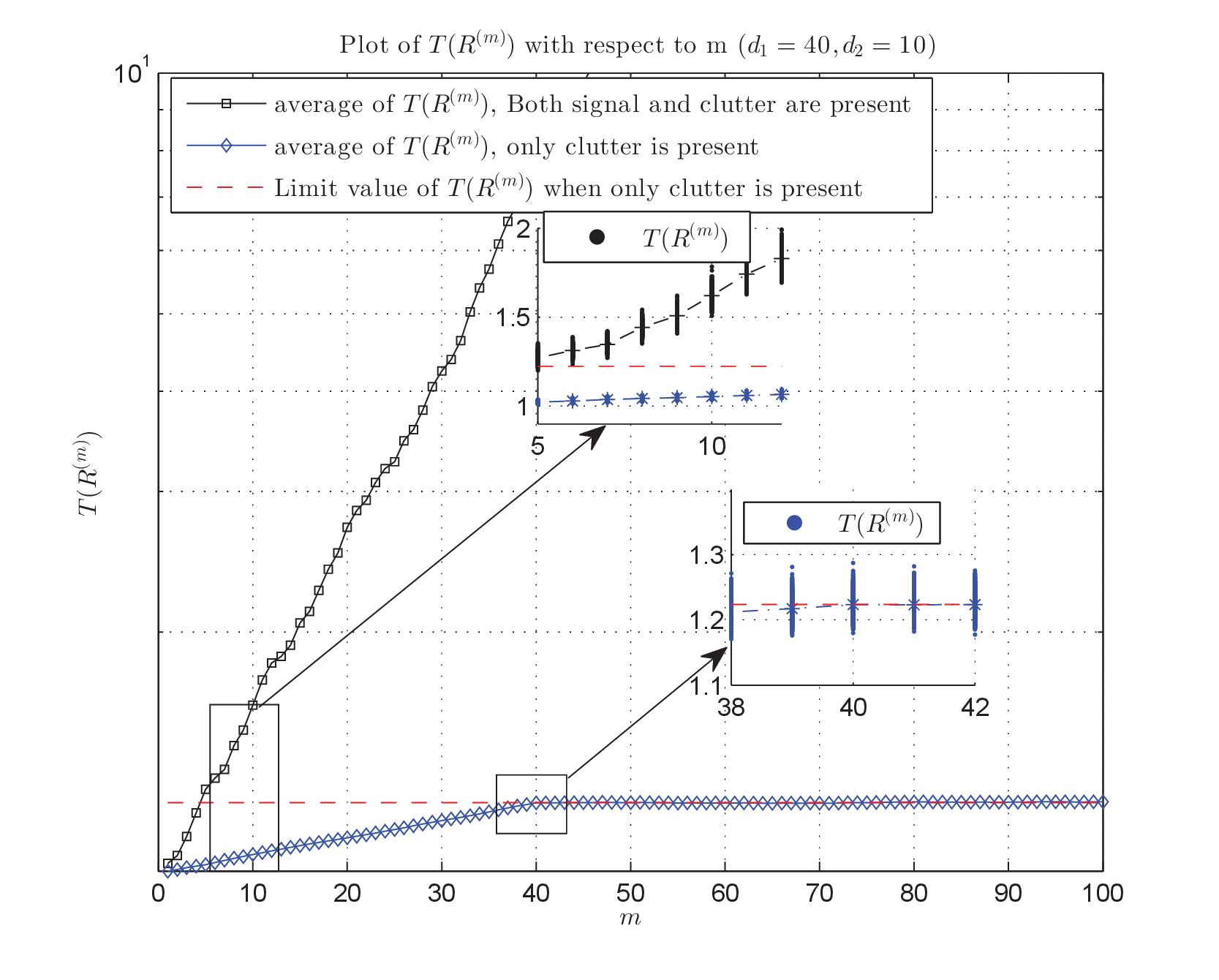}
	\caption{Simulation of the test quantity $T(\bm R^{(m)})$ in noisy environment}
	\label{figure3}
\end{figure}

\begin{figure}[htbp]
\centering
\includegraphics*[width=0.9\textwidth]{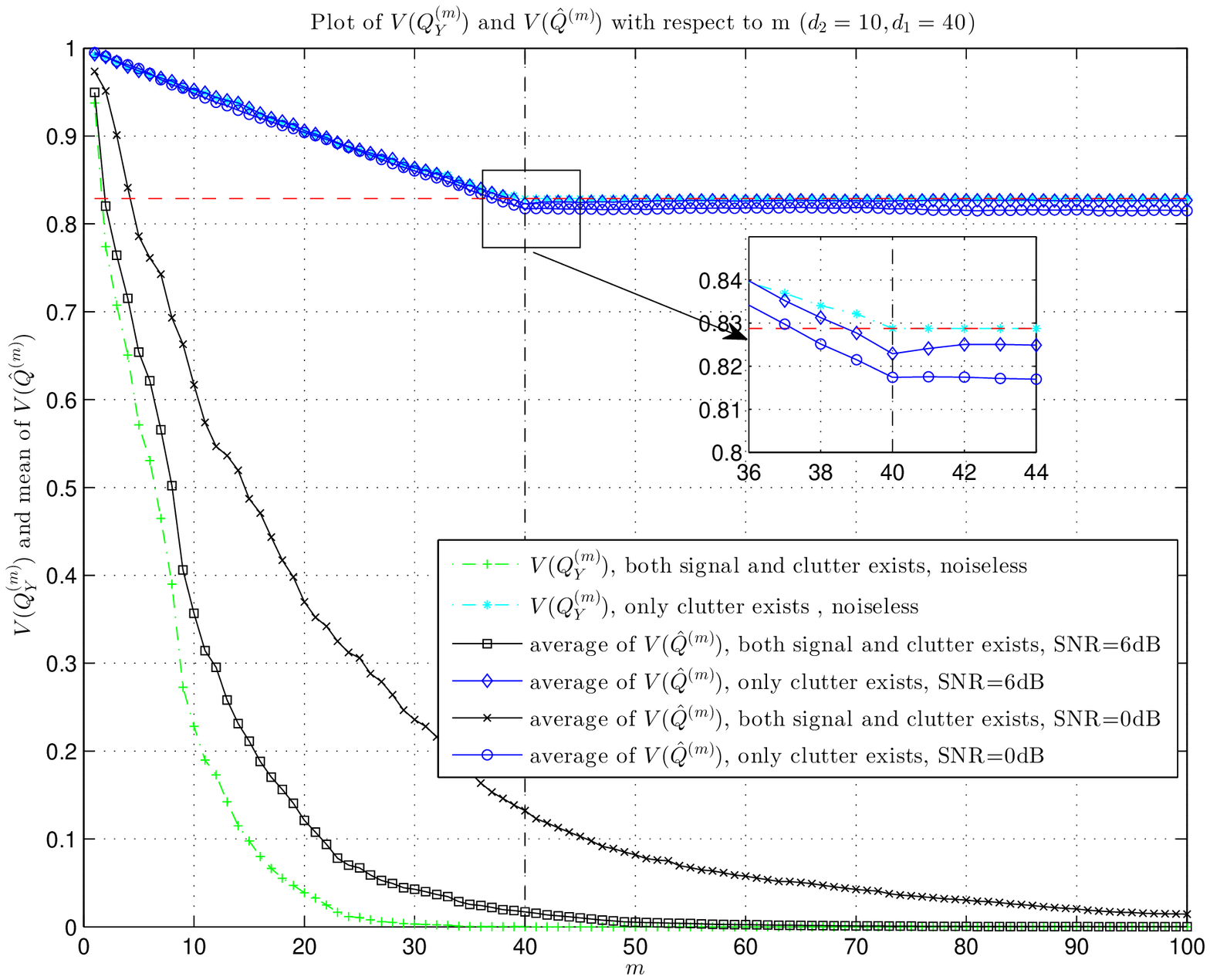}
\caption{Comparing the volume-based correlation $V(\hat{\bm Q}^{(m)})$ with respect to different SNR}
		\label{figure4}
\end{figure}

\section{Further Discussion \& Future Work}

\subsection{Volume and Inner Product}

Inner product is widely used for signal correlation under the theoretical framework of linear space, where certain signal sample is regarded as a point in this linear space. The matched filter is well-known for it has the optimal output SNR, and it is actually constructed based on inner-product. So inner product has long been regarded as a powerful tool for signal processing and received much attention. However, another geometrical quantity with close relationship with inner product still lies out of sight and receives little concern. That is volume, the focus of this paper.

The connection between inner product and volume could be seen from a simple fact about complex number multiplication. For two complex numbers $z_1=x_1+jy_1$, $z_2=x_2+jy_2$, where $j=\sqrt{-1}$, we have
\begin{equation}\label{cplx1}
z_1\overline{z_2}=(x_1x_2+y_1y_2)+j(x_1y_2-x_2y_1),
\end{equation}
Here the real part of (\ref{cplx1}) is just the inner product between vectors $\overrightarrow{x}=(x_1,y_1)$ and $\overrightarrow{y}=(x_2,y_2)$. How about the imaginary part? We have
\begin{equation}
x_1y_2-x_2y_1=\det\left(\begin{array}{cc}x_1&y_1\\x_2&y_2\end{array}\right)=\overrightarrow{x}\times\overrightarrow{y},
\end{equation}
It is exactly the outer (cross) product of $\overrightarrow{x}$ and $\overrightarrow{y}$, that is, volume of parallelogram spanned by them. Intuitively, the angle between two vectors could be introduced naturally using inner and outer product. The inner product represents cosine of angle and outer product means sine of angle. The judgement on the extent of collinearity between two vectors can be made based on inner product and cross product as well. The difference just lies at that the typical value of perfect collinearity is $1$ for inner product and $0$ for cross product. If we are considering two vectors, the volume and cross product coincides with each other occasionally. Therefore, inner product and volume are two sides of the same diamond.

In the multiple-vector case, i.e., when we are considering more than three vectors, the definition of cross product between multiple vectors becomes vague but that of volume is still making sense. To determine whether or not a vector $\overrightarrow{a}$ lies in a certain $d$-dimensional subspace $\bm H_S$ spanned by a group of vectors is more difficult because it doesn't work by calculating the inner products between $\overrightarrow{a}$ and each basis vector of subspace $\bm H_S$. Nevertheless, volume provides us a clue to solve this problem effectively. According to the volume's definition (\ref{VolumeDef1}), we only need to compute the $d+1$-dimensional volume of subspace spanned by $\overrightarrow{a}$ and basis vectors of $\bm H_S$ to inspect its vanishment. Volume can be used as a criterion for coplanarity, like the role inner-product plays for collinearity.

\subsection{Volume Correlation and Coherence}

In the recent plenary talk " Coherence as an Organizing Principle in Signal processing " in SSP'2014, Louis Scharf emphasized the importance of theory and application of " Coherence " between two random vectors, this "coherence" has been used in various signal processing applications \cite{Ramirez2010}\cite{klausner2014detection}. Consider random vectors $\bm{u}\in\mathbb{C}^p$ and $\bm{v}\in\mathbb{C}^q$, denote their cross covariance matrix by
\begin{equation}
\bm{R}=\left(\begin{array}{cc}\bm{R}_{\bm{uu}}&\bm{R}_{\bm{uv}}\\\bm{R}_{\bm{uv}}^{H}&\bm{R}_{\bm{vv}}\end{array}\right)
\end{equation}
the coherence is defined as
\begin{equation}\label{coh1}
\rho^2(\bm{u},\bm{v})=1-\frac{\displaystyle\det(\bm{R}_{\bm{uu}}-\bm{R}_{\bm{uv}}\bm{R}_{\bm{vv}}^{-1}\bm{R}_{\bm{uv}}^{H})}{\displaystyle\det(\bm{R}_{\bm{uu}})}
\end{equation}
On the other hand, the deterministic version of coherence called Euclidean coherence can also be defined for two rectangle matrix $\bm{U}\in\mathbb{C}^{n\times{p}}$ and $\bm{V}\in\mathbb{C}^{n\times{q}}$, which could be regarded as the samples of random vectors $\bm{u}$ and $\bm{v}$ respectively. Suppose Grammian matrix of $\bm{U}$ and $\bm{V}$ be
\begin{equation}
\bm{R}=\left(\begin{array}{cc}\bm{U}^{H}\bm{U}&\bm{U}^{H}\bm{V}\\\bm{V}^{H}\bm{U}&\bm{V}^{H}\bm{V}\end{array}\right)
\end{equation}
then the coherence may be defined as well,
\begin{equation}\label{coh2}
\hat{\rho}^2(\bm{U},\bm{V})=1-\frac{\displaystyle\det(\bm{U}^{H}\bm{U}-\bm{U}^{H}\bm{P}_{\bm{V}}\bm{U})}{\displaystyle\det(\bm{U}^{H}\bm{U})}
\end{equation}
where $\bm{P}_{\bm{V}}=\bm{V}(\bm{V}^{H}\bm{V})^{-1}\bm{V}^{H}$ is the projection matrix on subspace spanned by column of $\bm{V}$. Obviously, (\ref{coh2}) is one kind of finite sample estimation of (\ref{coh1}).

From fine-grain viewpoint, we can write Hilbert coherence as
\begin{align}
\rho^2(\bm{u},\bm{v})&=1-\frac{\displaystyle\det(\bm{R}_{\bm{uu}}-\bm{R}_{\bm{uv}}\bm{R}_{\bm{vv}}^{-1}\bm{R}_{\bm{uv}}^{H})}{\displaystyle\det(\bm{R}_{\bm{uu}})}\nonumber\\
&=1-\prod_{i=1}^{\min(p,q)}(1-k_i^2)
\end{align}
where $k_i$ are singular values of the coherence matrix $\bm{R}_{\bm{uu}}^{-1/2}\bm{R}_{\bm{uv}}\bm{R}_{\bm{vv}}^{-H/2}$. They are called canonical correlations \cite{scharf2000canonical} between the canonical coordinates of random variable $\bm{u}$ and $\bm{v}$. Euclidean coherence can be treated in the same way,
\begin{align}
\hat{\rho}^2(\bm{U},\bm{V})&=1-\frac{\displaystyle\det(\bm{U}^{H}\bm{U}-\bm{U}^{H}\bm{P}_{\bm{V}}\bm{U})}{\displaystyle\det(\bm{U}^{H}\bm{U})}\nonumber\\
&=1-\prod_{i=1}^{\min(p,q)}(1-\cos^2\theta_i(\Span(\bm{U}),\Span(\bm{V})))
\end{align}
where $\{\theta_i(\Span(\bm{U}), \Span(\bm{V})),1\leq{i}\leq\min(p,q)\}$ are the principal angles of subspace $\Span(\bm{U})$ and $\Span(\bm{V})$.

In view of (\ref{VolAng}), we have
\begin{align}
\hat{\rho}^2(\bm{U},\bm{V})&=1-\prod_{i=1}^{\min(p,q)}\sin^2\theta_i(\Span(\bm{U}),\Span(\bm{V}))\nonumber\\
&=1-(\Corr_{\textrm{vol}}(\Span(\bm{U}),\Span(\bm{V})))^2,
\end{align}
This is a interesting relationship between Euclidean coherence and the volume correlation (\ref{VolAng}) defined in this paper. Hence the potential value of volume correlation in signal processing field will be highlighted generally along with widespread of coherence, just as predicted by Scharf.

\subsection{The Phase Transition Phenomenon of random matrices and Future works}

In the regime of ultimate low SNR, It have been shown that significant errors will enter the estimation of signal subspaces obtained with the eigen-decomposition methods. This is due to the so-called Phase Transition Phenomenon on eigenvalues and eigenvectors of random matrices \cite{watkin1994optimal}\cite{paul2007asymptotics}\cite{paul2013random}. In fact, for a $n \times n$ sampled covariance matrix calculated in (\ref{CorrMat}), as $n\rightarrow\infty$ and $n/m\rightarrow\gamma\in(0,1)$, we have:
\par
(a) If $\lambda_i /\sigma^2 > 1+\sqrt{\gamma}$, then
\begin{eqnarray}
\hat{\lambda}^{(m)}_i &\stackrel{a.s.}{\rightarrow}& \lambda_i (1+\frac{\gamma\sigma^2}{\lambda_i-\sigma^2}),\quad i=1,\cdots k_m \\
|\langle \hat{\bm q}^{(m)}_i ,\bm q_i\rangle | &\stackrel{a.s.}{\rightarrow}&  \sqrt{\left( 1-\frac{\gamma\sigma^4}{(\lambda_i-\sigma^2)^2} \right)/\left(1+\frac{\gamma \sigma^2}{\lambda_i-\sigma^2} \right)},\quad i=1,\cdots k_m.
\end{eqnarray}
\par
(b) If $\lambda_i /\sigma^2 \leq 1+\sqrt{\gamma}$, then
\begin{eqnarray}
\hat{\lambda}^{(m)}_i &\stackrel{a.s.}{\rightarrow}& \sigma^2(1+\sqrt{\gamma})^2,\quad i=1,\cdots k_m \\
|\langle \hat{\bm q}^{(m)}_i ,\bm q_i\rangle | &\stackrel{a.s.}{\rightarrow}&  0,\quad i=1,\cdots k_m,
\end{eqnarray}
where  $\lambda_i$ (in descending order with increase of $i$) and $\bm q_i$ are eigenvalues and eigenvectors of the auto-correlation matrix of the received signal, and $\hat{\lambda}^{(m)}_i$ and $ \hat{\bm q}^{(m)}_i$ are the corresponding estimation of $\lambda_i$ and $\bm q_i$ calculated from eigenvalue decomposition of the sampled covariance matrix.

This phase transition phenomenon means that when $\lambda_i/\sigma^2$ was lower than a threshold, the estimated eigenvalues $\hat{\lambda}^{(m)}_i$ and eigenvectors $ \hat{\bm q}^{(m)}_i$ will resembles the eigenvalues and eigenvectors of the noise subspace almost surely. On the other hand, if the noise variance $\sigma^2$ was significantly high, such that $\lambda_i/\sigma^2$ is lower than $1+\sqrt{\gamma}$, then the estimated signal subspace $\Span(\hat{\bm Q}^{(m)})$ will not asymptotically converge to the true signal subspace $\Span(\bm Q_{SC})$ as expected. Therefore the performance of volume-correlation will degrade heavily. As a whole, the signal subspace estimation method in Detector 2 has its shortage for low SNR. The estimation of signal subspace in low SNR scenario remains an important future work for our volume-correlation subspace detector.

\section{Conclusion}

In this paper, we proposed a novel volume-correlation subspace detector and use it to treat the problem of subspace signal detection with noise and clutter. The proposed detector can effectively detect target signal in certain subspace by calculating the volume of parallelotope spanned by the basis of known target signal subspace together with the multi-dimensional observations of the received signal. It is shown theoretically that the detector can eliminate the influence of clutter while detecting without any knowledge on the clutter subspace, we call this unique property as " Detection while Learning ". Numerical simulations demonstrated the advantage and excellent performance of volume-correlation subspace detector.

\appendices

\section{Proof of Theorem \ref{H1Thm1}}

We will prove the monotone property (\ref{Monotone}). An useful lemma will be proved firstly. It is just a simple property with very clear geometric intuition for volume of subspaces.

\begin{Lemma}\label{Lemma1}
Let $\bm{Y}^{(m)}=[\bm{y}_1,\cdots,\bm{y}_m]\in\mathbb{R}^{n\times{m}}$ and $\bm{X}\in\mathbb{R}^{n\times{l}}$ be two matrices, then
\begin{equation}
  \Vol_{m+l}([\bm{X},\bm{Y}^{(m)}])=\Vol_{m-1+l}([\bm{X},\bm{Y}^{(m-1)}])\|\bm{P}_{[\bm{X},\bm{Y}^{(m-1)}]}^{\perp}\bm{y}_m\|
\end{equation}
where $\bm{P}_{\bm{A}}^{\perp}$ is the projection matrix onto the orthogonal complement of column space of matrix $\bm{A}$.
\end{Lemma}
\begin{proof}
  According to the definition of volume in (\ref{VolumeDef2}), we have
  \begin{align}
    \Vol_{m+l}^2([\bm{X},\bm{Y}^{(m)}])&=\det([\bm{X},\bm{Y}^{(m)}]^{T}[\bm{X},\bm{Y}^{(m)}])\nonumber\\
    &=\det\left(
           \begin{array}{cc}
           \bm{X}^{T}\bm{X}&\bm{X}^{T}\bm{Y}^{(m)}\\
           (\bm{Y}^{(m)})^T\bm{X}&(\bm{Y}^{(m)})^T\bm{Y}^{(m)}
           \end{array}
           \right)\nonumber\\
    &=\det\left(
           \begin{array}{cc:c}
           \bm{X}^{T}\bm{X}&\bm{X}^{T}\bm{Y}^{(m-1)}&\bm{X}^{T}\bm{y}^{(m)}\\
           (\bm{Y}^{(m-1)})^T\bm{X}&(\bm{Y}^{(m-1)})^T\bm{Y}^{(m-1)}&(\bm{Y}^{(m-1)})^T\bm{y}_m\\
           \hdashline
           \bm{y}_m^T\bm{X}&\bm{y}_m^T\bm{Y}^{(m-1)}&\bm{y}_m^T\bm{y}_m\\
           \end{array}
           \right)
  \end{align}
  Using Schur complement formula,
  \begin{equation}
    \det\left(
        \begin{array}{cc}
        \bm{A}&\bm{B}\\
        \bm{C}&\bm{D}
        \end{array}
        \right)
    =\det(\bm{A})\det(\bm{D}-\bm{B}\bm{A}^{-1}\bm{C}),
  \end{equation}
  we obtain
  \begin{align}
    \Vol_{m+l}^2([\bm{X},\bm{Y}^{(m)}])
    =&\det\left(
          \begin{array}{cc}
          \bm{X}^{T}\bm{X}&\bm{X}^{T}\bm{Y}^{(m-1)}\\
          (\bm{Y}^{(m-1)})^T\bm{X}&(\bm{Y}^{(m-1)})^T\bm{Y}^{(m-1)}
          \end{array}
          \right)\nonumber\\
      &\det\left(
          \bm{y}_m^T\bm{y}_m-\bm{y}_m^T[\bm{X},\bm{Y}^{(m-1)}]
           \left(
            \begin{array}{cc}
            \bm{X}^{T}\bm{X}&\bm{X}^{T}\bm{Y}^{(m-1)}\\
            (\bm{Y}^{(m-1)})^T\bm{X}&(\bm{Y}^{(m-1)})^T\bm{Y}^{(m-1)}
            \end{array}
           \right)^{-1}
           [\bm{X},\bm{Y}^{(m-1)}]^T\bm{y}_m
          \right)\nonumber\\
    =&\Vol_{m-1+l}^2([\bm{X},\bm{Y}^{(m-1)}])\bm{y}_m^{T}(\bm{I}-\bm{P}_{[\bm{X},\bm{Y}^{(m-1)}]})\bm{y}_m,
  \end{align}
  Because $\bm{I}-\bm{P}_{[\bm{X},\bm{Y}^{(m-1)}]}$ is an idempotent matrix,
  \begin{equation}
    (\bm{I}-\bm{P}_{[\bm{X},\bm{Y}^{(m-1)}]})^2=\bm{I}-\bm{P}_{[\bm{X},\bm{Y}^{(m-1)}]},
  \end{equation}
  we obtain
  \begin{align}
    \Vol_{m+l}^2([\bm{X},\bm{Y}^{(m)}])
    &=\Vol_{m-1+l}^2([\bm{X},\bm{Y}^{(m-1)}])\bm{y}_k^{T}(\bm{I}-\bm{P}_{[\bm{X},\bm{Y}^{(m-1)}]})^2\bm{y}_m\nonumber\\
    &=\Vol_{m-1+l}^2([\bm{X},\bm{Y}^{(m-1)}])\|\bm{P}_{[\bm{X},\bm{Y}^{(m-1)}]}^{\perp}\bm{y}_m\|^2
  \end{align}
  This is just what we want to prove.
\end{proof}
\par

With the notations in theorem \ref{H1Thm1}, when $m<\Dim(\bm{H}_S\oplus\bm{H}_C)$ (or $m<\Dim(\bm{H}_C)$), we have
\begin{equation}
V(\bm{Y}^{(m)})=\Vol_{m+d_2}([\bm{Q}_{S}, \bm{Q}_{\bm{Y}}^{(m)}]),
\end{equation}
where $\bm{Q}_{\bm{Y}}^{(m)}=[\bm{q}_{\bm{Y}}^{(1)},\bm{q}_{\bm{Y}}^{(2)},\cdots,\bm{q}_{\bm{Y}}^{(m)}]$ is the orthogonal basis matrix of $\Span(\bm{Y}^{(m)})$. Moreover, if we write
\begin{equation}
\bm{Q}_{\bm Y^{(m)}}=[\bm{Q}_{\bm Y^{(m-1)}}, \bm{q}_{\bm{Y}}^{(m)}],
\end{equation}
then according to the way we construct $\bm{Q}_{\bm{Y}}^{(m)}$, we have
\begin{align}
\bm{q}_{\bm{Y}}^{(m)}
&=\frac{(\bm{I}-\bm{Q}_{\bm{Y}}^{(m-1)}(\bm{Q}_{\bm{Y}}^{(m-1)})^T)\bm{y}_{m}}
       {\|(\bm{I}-\bm{Q}_{\bm{Y}}^{(m-1)}(\bm{Q}_{\bm{Y}}^{(m-1)})^T)\bm{y}_{m}\|}\nonumber\\
&=\frac{\bm{P}_{\bm{Q}_{\bm{Y}}^{(m-1)}}^{\perp}\bm{y}_{m}}{\|\bm{P}_{\bm{Q}_{\bm{Y}}^{(m-1)}}^{\perp}\bm{y}_{m}\|}.
\end{align}
Using Lemma \ref{Lemma1}, Let $\bm{X}=\bm{Q}_S$, $\bm{Y}=\bm{Q}_{\bm{Y}}^{(m)}$, we obtain
\begin{align}
  V(\bm{Y}^{(m)})&=\Vol_{d_2+m-1}([\bm{Q}_{S},\bm{Q}_{\bm{Y}}^{(m-1)}])\|\bm{P}_{[\bm{Q}_S,\bm{Q}_{\bm{Y}}^{(m-1)}]}^{\perp}\bm{q}_{\bm{Y}}^{(m)}\|\\
  &=V(\bm{Y}^{(m-1)})\|\bm{P}_{[\bm{Q}_S,\bm{Q}_{\bm{Y}}^{(m-1)}]}^{\perp}\bm{q}_{\bm{Y}}^{(m)}\|.
\end{align}
Take into account the property of projection matrices,
\begin{equation}
  \|\bm{P}_{[\bm{Q}_S,\bm{Q}_{\bm{Y}}^{(m-1)}]}^{\perp}\bm{q}_{\bm{Y}}^{(m)}\|\leq\|\bm{q}_{\bm{Y}}^{(m)}\|=1,
\end{equation}
we have
\begin{equation}
  V(\bm{Y}^{(m)})\leq{V(\bm{Y}^{(m-1)})}.
\end{equation}
Thus (\ref{Monotone}) and theorem \ref{H1Thm1} has been proven.

\section{Proof of Theorem \ref{H0Thm2}}

Firstly, we prove the necessary part of this theorem. 
Let $\bm{H}_S$ and $\bm{H}_C$ be target subspace and clutter subspace respectively. $d_2 = \Dim(\bm{H}_S)$, $d_1=\Dim(\bm{H}_C)$. Assume there exists target signal in received data $\{\bm{y}_i, i=1,\cdots,m\}$, that is to say,
\begin{equation}
\bm{y}_i=\bm s_i+\bm c_i,\quad \bm s_i \in \bm H_S, \bm c_i \in \bm H_C
\end{equation}
for some $i\in{S}\subset\{1,2,\cdots,m\}$, and
\begin{equation}
\bm{y}_i=\bm c_i,
\end{equation}
for other $i\in\{1,2,\cdots,m\}\setminus{S}$. Then under the generic hypothesis, we have
\begin{equation}
  \Rank([\bm{y}_1,\cdots,\bm{y}_m])=m,
\end{equation}
for $m\leq{d_1}$. Hence the result of successive orthogonalization could be written as
\begin{equation}
  \bm{Q}_{\bm{Y}}^{(m)}=[\bm{q}_{\bm{Y}}^{(1)},\bm{q}_{\bm{Y}}^{(2)},\cdots,\bm{q}_{\bm{Y}}^{(m)}],
\end{equation}
It should be stressed that in the case of $m=d_1+1$, we still have
\begin{equation}
  \Rank([\bm{y}_1,\cdots,\bm{y}_{d_1+1}])=d_1+1,
\end{equation}
because of the presence of target signal. In other words,
\begin{equation}\label{Ortho1}
  \bm{Q}_{\bm{Y}}^{(d_1+1)}=[\bm{q}_{\bm{Y}}^{(1)},\bm{q}_{\bm{Y}}^{(2)},\cdots,\bm{q}_{\bm{Y}}^{(d_1)},\bm{q}_{\bm{Y}}^{(d_1+1)}].
\end{equation}
For $\bm{Y}^{(d_1+1)}$, all of its $d_1+1$ linearly independent directions includes two parts, one with $d_1$ directions coming from clutter subspace $\bm{H}_C$ and the other one direction contributed by target subspace $\bm{H}_S$.

Let $\bm{Q}_S$ be the matrix with columns being the orthonormal basis vectors of $\bm{H}_S$,  (\ref{Ortho1}) means that
\begin{equation}
  \Rank([\bm{Q}_S,\bm{Q}_{\bm{Y}}^{(d_1+1)}])=d_2+d_1,
\end{equation}
but the number of nonzero columns of $[\bm{Q}_S,\bm{Q}_{\bm{Y}}^{(d_1+1)}]$ is $d_1+d_2+1$, so using Lemma \ref{lemma2}, we obtain
\begin{equation}
  \Vol_{d_1+d_2+1}([\bm{Q}_S,\bm{Q}_{\bm{Y}}^{(d_1+1)}])=0,
\end{equation}
Take $K=d_1+1$, the necessary part of theorem has been proved.

On the contrary, under the generic hypothesis, we need to prove that, if there exists $K$ such that
\begin{equation}\label{SigEx}
  \Vol_{K+d_2}([\bm{Q}_S,\bm{Q}_{\bm{Y}}^{(K)}])=0,
\end{equation}
then there must be target signal in sample data ${\bm{y}_1,\cdots,\bm{y}_K}$.

Assume this was not the case, then every sample $\bm{y}_i$ contains no target signal, which means for $\forall m \in \mathbb{N}$
\begin{equation}
  \Span(\bm{Y}^{(m)})=\Span\{\bm{y}_1,\cdots,\bm{y}_m\}\subset\bm{H}_C,
\end{equation}
therefore the orthonormal basis matrix $\bm{Q}_{\bm{Y}}^{(m)}$ of $\bm{Y}^{{m}}$ satisfies
\begin{equation}
  \bm{Q}_{\bm{Y}}^{(m)}=[\bm{q}_{\bm{Y}}^{(1)},\bm{q}_{\bm{Y}}^{(2)},\cdots,\bm{q}_{\bm{Y}}^{(m)}]
\end{equation}
for $m\leq{d_1}$ and
\begin{equation}
  \bm{Q}_{\bm{Y}}^{(m)}=[\bm{q}_{\bm{Y}}^{(1)},\bm{q}_{\bm{Y}}^{(2)},\cdots,\bm{q}_{\bm{Y}}^{(d_1)}]
\end{equation}
for $m>d_1$. When $m > d_1$, according to (\ref{CorrVol}) and (\ref{VolAng}), we obtain
\begin{align}
  \Vol_{d_1+d_2}([\bm{Q}_S,\bm{Q}_{\bm{Y}}^{(m)}])
  &=\frac{\Vol_{d_1+k_2}([\bm{Q}_S,\bm{Q}_{\bm{Y}}^{(m)}])}{\Vol_{d_2}(\bm{Q}_S)\Vol_{d_1}(\bm{Q}_{\bm{Y}}^{(m)})}\nonumber\\
  &=\Corr_{\textrm{vol}}(\bm{Q}_S,\bm{Q}_{\bm{Y}}^{(m)})\nonumber\\
  &=\prod_{j=1}^{\min(d_1,d_2)}\sin\theta_j(\bm{H}_S, \bm{H}_C)\nonumber\\
  &>0.
\end{align}
Considering the monotone relation (\ref{Monotone}), we have
\begin{equation}
  \Vol_{m+d_2}([\bm{Q}_S,\bm{Q}_{\bm{Y}}^{(m)}])>0,\quad{\forall m\in\mathbb{N}},
\end{equation}
Contradiction! We have verified the sufficient part and the whole theorem has been proved.

\section{Proof of Theorem \ref{H1Thm1N}}

To complete the proof of Theorem \ref{H1Thm1N}, several lemmas are required as necessary tools. These lemmas concerned with asymptotical distribution of eigenvectors of sample covariance matrix, concentration bounds and matrix perturbation.

\begin{Lemma}\cite{stoica1989music}\label{stoica}
Consider the matrix $\bm{\hat{Q}}^{(m)}\in\mathbb{R}^{n\times{r}}$ with columns being the $r$ eigenvectors of sample covariance matrix $\hat{\mathrm{R}}^{(m)}\in\mathbb{R}^{n\times{n}}$ calculated from $m$ samples, corresponding to the largest $r$ eigenvalues, i.e.,
$$
\bm{\hat{Q}}^{(m)}=[\bm{\hat{q}}_1^{(m)}, \bm{\hat{q}}_2^{(m)},\cdots,\bm{\hat{q}}_r^{(m)}],
$$
then $\bm{\hat{Q}}^{(m)}$'s asymptotic distribution (for large $m$) is jointly Gaussian with mean
$$
\bm{Q}=[\bm{q}_1,\bm{q}_2,\cdots,\bm{q}_r],
$$
and covariance $\bm \Sigma_1^{(m)},\cdots,\bm \Sigma_r^{(m)}$, where
\begin{equation}
\bm \Sigma_i^{(m)}:= \frac{\lambda_i}{m}\Big[\sum_{\stackrel{j=1}{j \neq i}}^r\frac{\lambda_j}{(\lambda_i-\lambda_j)^2}\bm{q}_j\bm{q}_j^T + \sum_{j=r+1}^{P}\frac{\sigma^2}{(\sigma^2-\lambda_i)^2}\bm{q}_j\bm{q}_j^T\Big],\quad i=1,\cdots,r
\end{equation}
and
\begin{equation}\label{perturbcov}
\mathbb{E}(\bm{\hat{q}}_i^{(m)}-\bm{q}_i)(\bm{\hat{q}}_k^{(m)}-\bm{q}_k)^T=\bm\Sigma_i^{(m)}\cdot\delta_{i,k},\quad i,k=1,\cdots,r
\end{equation}
where $\lambda_1\geq\lambda_2\geq\cdots\lambda_r\geq\lambda_{r+1}=\cdots=\lambda_n=\sigma^2$ are eigenvalues of the covariance matrix $\mathrm{R}_r$ in (\ref{CorrMatrix}), with $\bm{q}_1,\cdots,\bm{q}_n$ the corresponding eigenvectors.
\end{Lemma}

\begin{Lemma}\label{LemmaRMF}
For the random matrix
\begin{equation}
\bm{E} = [\bm{e}_1,\cdots,\bm{e}_r]\in\mathbb{R}^{n\times{r}}
\end{equation}
where $\bm{e}_i\thicksim\mathcal{N}(0,\bm\Sigma_i), 1\leq{i}\leq{d}$, and $\mathbb{E}(\bm{e}_i\bm{e}_k^T)=\bm\Sigma_i\cdot\delta_{i,k}$, then for any
$0<\varepsilon<1$, there exists a constant $C>0$ that depends on $\bm\Sigma_i$, such that
\begin{equation}\label{concentration}
\|\bm{E}\|_F^2\leq(1+\varepsilon)\sum_{i=1}^r\Trace(\bm\Sigma_i),
\end{equation}
holds with probability
\begin{equation}\label{prob}
\mathbb{P}\geq1-\exp\{-\frac{r\cdot{n}\cdot\varepsilon^2}{C}\}.
\end{equation}
\end{Lemma}

\begin{proof}
From the definition of Frobenius norm, we know that
\begin{equation}\label{FNorm}
\|\bm{E}\|_F^2 = \sum_{i=1}^r\|\bm{e}_i\|_2^2.
\end{equation}
For any $1\leq{i}\leq{r}$, $\bm{e}_i\thicksim\mathcal{N}(0,\bm\Sigma_i)$, the eigenvalue decomposition of $\bm\Sigma_i\in\mathbb{R}^{n\times{n}}$ could be written as
\begin{equation}
\bm\Sigma_i=\bm{V}_i\bm\Lambda_i\bm{V}_i^T,
\end{equation}
where the diagonal matrix $\bm\Lambda_i:=\Diag(\sigma_{i,1}^2,\cdots,\sigma_{i,n}^2)$ and $\sigma_{i,1}^2\geq\sigma_{i,2}^2\geq\cdots\geq \sigma_{i,n}^2\geq0$ are eigenvalues of $\bm\Sigma_i$.

Let
\begin{equation}
\bm{\tilde e}_i=\bm{V}_i^T\bm{e}_i,
\end{equation}
then
\begin{equation}
\bm{\tilde e}_i\thicksim\mathcal{N}(0,\bm\Lambda_i),\quad \|\bm{\tilde e}_i \|_2^2=\|\bm{e}_i\|_2^2.
\end{equation}
Denote the elements of vector $\bm{\tilde e}_i$ by
\begin{equation}
\bm{\tilde e}_i = [\tilde e_{i,1},\cdots,\cdots \tilde{e}_{i,n}]^T,
\end{equation}
then different $\tilde e_{i,j}$ are independent and satisfy
\begin{equation}\label{elemdist}
\tilde e_{i,j} \thicksim \mathcal{N}(0,\sigma_{i,j}^2),\quad 1\leq j \leq n.
\end{equation}
Now we stack all these vectors $\bm{\tilde e}_i, 1 \leq i \leq r$ into a single vector, i.e., we let
\begin{equation}
\bm{\tilde e} := [\bm{\tilde e}_1^T, \bm{\tilde e}_2^T, \cdots, \bm{\tilde e}_r^T]^T \in \mathbb{R}^{n\cdot{r}},
\end{equation}
then we have
\begin{equation}
\bm{\tilde e} \thicksim \mathcal{N}(0, \bm \Lambda), \quad \bm \Lambda = \Diag(\bm \Lambda_1, \cdots, \bm \Lambda_r) = \Diag(\sigma_{1,1}^2,\cdots,\sigma_{1,n}^2,\cdots,\sigma_{r,1}^2,\cdots,\sigma_{r,n}^2).
\end{equation}
Therefore, (\ref{FNorm}) is equivalent to
\begin{equation}
\|\bm E\|_F^2 = \sum_{i=1}^r  \|\bm{\tilde e}_i\|_2^2 = \|\bm{\tilde e}\|_2^2.
\end{equation}

It is well known that the norm of a Gaussian random vector will concentrate around its expectation \cite{ledoux2001concentration}. It has been proved that the norm of an i.i.d. Gaussian random vector will concentrate around its expectation (Chapter 4, \cite{PHDThesis}). The problem of the concentration of $\|\bm{\tilde e}\|_2^2$ here is only slightly different with the one in \cite{PHDThesis}. In particular, the elements of
$\bm{\tilde e}$ have different variances in our case. Therefore, the proof will be adapted from the proof of Theorem 4.2 in \cite{PHDThesis}. So only the different part will be given in the following proof.

Firstly, we have
\begin{equation}
\mathbb{E}\{\|\bm{\tilde e}\|_2^2\} = \sum_{i=1}^r\sum_{j=1}^n \sigma_{i,j}^2 = \sum_{i=1}^r\Trace(\bm \Sigma_i),
\end{equation}

Then follow the same approach as \cite{PHDThesis} and utilize Markov's Inequality. For any parameter $\beta>0$ and $\lambda>0$, we have
\begin{eqnarray}
\mathbb{P}\{\|\bm{\tilde e}\|_2^2 \geq \beta \sum_{i=1}^r \Trace(\bm \Sigma_i)\} &=& \mathbb{P}\{\exp(\lambda \|\bm{\tilde e}\|_2^2) \geq \exp(\lambda\beta \sum_{i=1}^r \Trace(\bm \Sigma_i)) \} \nonumber \\
&=& \prod_{i=1}^r \mathbb{P}\{\exp(\lambda \|\bm{\tilde e}_i\|_2^2) \geq \exp(\lambda\beta  \Trace(\bm \Sigma_i)) \} \nonumber \\
&\leq& \prod_{i=1}^r  \frac{\mathbb{E}\{\exp(\lambda \|\bm{\tilde e}_i\|_2^2)\}}{\exp(\lambda\beta \Trace(\bm \Sigma_i))} \nonumber \\
&=&  \prod_{i=1}^r \prod_{j=1}^n [\frac{\mathbb{E}\{\exp(\lambda \tilde e_{i,j}^2)\}}{\exp(\lambda\beta \sigma_{i,j}^2)}].
\end{eqnarray}
The moment generating function of the Gaussian random variable $\tilde e_{i,j}$ is:
\begin{equation}
\mathbb{E}\{\exp(\lambda \tilde e_{i,j}^2)\}=\frac{1}{\sqrt{1-2\lambda \sigma_{i,j}^2}},
\end{equation}
let
\begin{equation}
\sigma_{\max}:=\max_{i,j}\sigma_{i,j},\quad \sigma_{\min} := \min_{i,j}\sigma_{i,j},
\end{equation}
we have
\begin{equation}\label{Obj}
\mathbb{P}\{\|\bm{\tilde e}\|_2^2 \geq \beta \sum_{i=1}^r\Trace(\bm \Sigma_i)\}\leq\left(\frac{\exp(-2\lambda\beta \sigma_{\min}^2)}{1-2\lambda\sigma_{\max}^2}\right)^{r\cdot{n}/2},\quad\lambda>0,
\end{equation}

The rest of proof is the same as Theorem 4.2 in \cite{PHDThesis} and will be described briefly. Replacing $\lambda$ with its optimal value such that the right side of (\ref{Obj}) is minimized, and regarding some  formulas involving $\sigma_{\max}$ and $\sigma_{\min}$ for a constant $C$, we can derive the result of this lemma (which is also the result of Corollary 4.1 in \cite{PHDThesis} under i.i.d hypothesis):
\begin{eqnarray}
\mathbb{P}\{\|\bm{\tilde e}\|_2^2 \geq (1+\varepsilon)\sum_{i=1}^r\Trace(\bm \Sigma_i)\}\leq\exp (-\frac{r\cdot{n}\cdot\varepsilon^2}{C}), 
\end{eqnarray}
holds for any $0<\varepsilon<1$, where $C>0$ is a constant depending on $\sigma_{\max}$ and $\sigma_{\min}$.
\end{proof}

Next, the lemma will be presented to estimate the influence of the error between sample eigenvectors $\bm{\hat{Q}}^{(m)}$ and its true value on the volume-correlation computation. Motivated by the relation between volume and determinant, the matrix perturbation theory was utilized to derive the result needed.

\begin{Lemma}(Corollary 2.7 in \cite{ipsen2008perturbation})\label{perturb}
For the matrix $\bm A \in \mathbb{R}^{n \times n}$, and the perturbation matrix $\bm E \in \mathbb{R}^{n \times n}$, we have
\begin{itemize}
\item If $\bm A $ is full-rank, then
\begin{equation}\label{H0perturb}
|\det(\bm A+\bm E)-\det(\bm A)| \leq \sum_{i=1}^n s_{n-i}(\bm A)\|\bm E\|_2^i,
\end{equation}
\item If $\Rank(\bm A) = k$ for some $1\leq k \leq n-1$, then
\begin{equation}\label{H1perturb}
|\det(\bm A+\bm E)| \leq \|\bm E\|_2^{n-k} \sum_{i=0}^k s_{k-i}(\bm A)\|\bm E\|_2^i.
\end{equation}
here $s_{k}(\bm A)$ is defined as the $k$th elementary symmetric function of singular values of matrix $\bm{A}\in\mathbb{R}^{n\times{n}}$:
\begin{equation}\label{ESF}
s_{k}(\bm A) := \sum_{1 \leq i_1 \leq \cdots \leq i_k \leq n} \sigma_{i_1}\cdots \sigma_{i_k}, 1\leq k \leq n.
\end{equation}
\end{itemize}
\end{Lemma}

Now we will prove theorem \ref{H1Thm1N}. Assume the sample data be
\begin{equation}
\bm{R}^{(m)}=[\bm{r}_1,\cdots,\bm{r}_m],
\end{equation}
and its sample correlation matrix be $\mathrm{R}_r^{(m)}=\frac{1}{m}\bm{R}^{(m)}(\bm{R}^{(m)})^T$, then the volume-correlation is of the form of:
\begin{equation}
V(\bm{R}^{(m)})=\Vol_{d_2+k_m}([\bm{Q}_S,\bm{\hat{Q}}^{(m)}]),
\end{equation}
where $\bm{Q}_S$ is the matrix with columns being the orthonormal basis vectors of target subspace $\bm{H}_S$ and $\bm{\hat{Q}}^{(m)}$ is the matrix with columns being the eigenvectors of $\mathrm{R}_r^{(m)}$ corresponding to its $k_m$ relatively large eigenvalues. Without loss of generality, let
\begin{equation}
  \bm{\hat{Q}}^{(m)}=[\bm{\hat{q}}_1,\cdots,\bm{\hat{q}}_{k_m}],
\end{equation}
We have
\begin{align}\label{detperturb}
V(\bm R^{(m)})&= \Vol_{d_2+k_m}([\bm{Q}_{S}, \bm{\hat{Q}}^{(m)}])\nonumber\\
&={\det}^{1/2}\left[\begin{array}{cc}
						\bm Q_{S}^T\bm Q_{S} & \bm Q_{S}^T\bm{\hat{Q}}^{(m)}\\
						(\bm{\hat{Q}}^{(m)})^T\bm Q_{S} & (\bm{\hat{Q}}^{(m)})^T\bm{\hat{Q}}^{(m)}
					\end{array}\right]\nonumber \\
&= {\det}^{1/2}\left(\bm Q_{S}^T\bm Q_{S}\right)\cdot
{\det}^{1/2}(\bm I_n - (\bm{\hat{Q}}^{(m)})^T\bm Q_{S}\left(\bm{Q}_{S}^T\bm{Q}_{S}\right)^{-1}\bm{Q}_{S}^T\bm{\hat{Q}}^{(m)})\nonumber\\
&= {\det}^{1/2}\left((\bm{\hat{Q}}^{(m)})^T\bm P_{S}^{\perp}\bm{\hat{Q}}^{(m)}\right).
\end{align}
where $\bm{P}_{S}^{\perp}$ is the projection matrix onto the orthogonal complement of target subspace $\bm{H}_S$.

We noticed that we need the approximation result of $\bm{\hat{Q}}^{(m)}$ according to lemma \ref{stoica} to obtain the conclusion of theorem \ref{H1Thm1N}. It is natural for using lemma \ref{stoica} to estimate the error of approximation. 
Firstly we express $\bm{\hat{Q}}^{(m)}$ as the following linear random perturbation model,
\begin{equation}
\bm{\hat{Q}}^{(m)} = \bm{Q}_{SC}+\bm{E}^{(m)},
\end{equation}
where $\bm{Q}_{SC}$ denotes the real orthogonal basis of signal subspace from $\mathrm{R}_{\bm r}$, and
\begin{equation}
 \bm E^{(m)} = [\bm e_1^{(m)},\cdots,\bm e_{k_m}^{(m)}],
\end{equation}
with $\bm e_i^{(m)}\thicksim \mathcal{N}(0,\bm \Sigma_i^{(m)})$, and ${\bm e_i^{(m)}}$ are mutually independent for different $i$. Then according to (\ref{perturbcov}), we have
\begin{equation}
\bm \Sigma_i^{(m)} = \frac{\lambda_i}{m} \Big[ \sum_{\stackrel{j=1}{j \neq i}}^{k_m} \frac{\lambda_j}{(\lambda_i-\lambda_j)^2} \bm q_j \bm q_j^T + \sum_{j=k_m+1}^{n}\frac{\sigma^2}{(\sigma^2-\lambda_i)^2}\bm q_j \bm q_j^T \Big].
\end{equation}
Therefore (\ref{detperturb}) becomes
\begin{align}\label{detperturb2}
V(\bm R^{(m)})
&={\det}^{1/2}\left((\bm{\hat{Q}}^{(m)})^T\bm P_{S}^{\perp}\bm{\hat{Q}}^{(m)}\right) \nonumber \\
&={\det}^{1/2}\left((\bm P_{S}^{\perp}\bm{Q}_{SC}+\bm P_{S}^{\perp}\bm E^{(m)})^T(\bm P_{S}^{\perp}\bm{Q}_{SC}+\bm P_{S}^{\perp}\bm E^{(m)})\right),
\end{align}
for simplicity, let
$$
\bm V = \bm P_{S}^{\perp}\bm{Q}_{SC}, \qquad \bm W = \bm P_{S}^{\perp}\bm E^{(m)},
$$
then (\ref{detperturb2}) becomes
\begin{equation}
T(\bm R^{(m)}) = \det\left((\bm V+ \bm W)^T(\bm V+ \bm W)\right),
\end{equation}
Let
$$
\bm A = \bm V^T \bm V, \qquad \bm E = \bm V^T \bm W + \bm W^T \bm V + \bm W^T \bm W,
$$
then we have
$$
V^2(\bm R^{(m)}) = \det(\bm A + \bm E),
$$
where
\begin{align}\label{RankA}
\bm{A}&=(\bm{Q}_{SC})^T\bm{P}_S^{\perp}\bm{Q}_{SC},\\
\bm{E}&=(\bm{Q}_{SC})^T\bm{P}_S^{\perp}\bm{E}^{(m)}+(\bm{E}^{(m)})^T\bm{P}_S^{\perp}\bm{Q}_{SC}+(\bm E^{(m)})^T\bm{P}_S^{\perp}\bm E^{(m)},
\end{align}

From lemma \ref{LemmaRMF}, we noted that its two conclusions were distinguished by the rank of matrix $\bm{A}$. It indicated that the rank of $\bm{A}$ was a critical factor for accuracy of approximation. In fact, it determined the infinitesimal order for error of approximation. So the rank of $\bm{A}$ should be analyzed.

According to (\ref{RankA}), we have
\begin{equation}\label{RankB}
  (\bm{Q}_{SC})^T\bm{P}_S^{\perp}\bm{Q}_{SC}=(\bm{Q}_{SC})^T(\bm{I}_n-\bm{P}_S)\bm{Q}_{SC}=\bm{I}_{k_m}-(\bm{Q}_{SC})^T\bm{P}_S\bm{Q}_{SC}.
\end{equation}
The rank of $\bm{A}$ is closely related to the rank of $(\bm{Q}_{SC})^T\bm{P}_S\bm{Q}_{SC}$, or more precisely, the structure of subspace spanned by $\bm{Q}_{SC}$. There are two possibilities for the structure of $\Span(\bm{Q}_{SC})$,
\begin{itemize}
\item If the target signal presents, then $\Dim(\Span(\bm{Q}_{SC})\cap\bm{H}_S)\neq0$ and $\Dim(\Span(\bm{Q}_{SC})\cap\bm{H}_C)\neq0$;
\item If the target signal doesn't present, then $\Dim(\Span(\bm{Q}_{SC})\cap\bm{H}_S)=0$ and $\Dim(\Span(\bm{Q}_{SC})\cap\bm{H}_C)\neq0$;
\end{itemize}

In the first case, because $\Span(\bm{Q}_{SC})\cap\bm{H}_S\neq\{0\}$, it is assumed that $k_m^S$ of $k_m$ columns of $\bm{Q}_{SC}$ were contributed by target subspace and the others came from clutter subspace. The corresponding result on the rank of $(\bm{Q}_{SC})^T\bm{P}_S^{\perp}\bm{Q}_{SC}$ can be summarized in the following lemmas
\begin{Lemma}
  Under the hypothesis of existence of target signal in sample data, we have
  \begin{equation}\label{Rank1}
    \Rank((\bm{Q}_{SC})^T\bm{P}_S^{\perp}\bm{Q}_{SC})=k_m-k_m^S,
  \end{equation}
  and all of its nonzero singular values (eigenvalues) are $1$.
\end{Lemma}
\begin{proof}
It should be noted firstly that $k_m$ can not excess $d_1+d_2$ which is the intrinsic dimension of $\bm{H}_S\oplus\bm{H}_C$, and $k_m^S$ should not be larger than $d_2$ no matter how large the number $m$ of sample data is. Generically, when $m$ is sufficiently large, which is the case we consider here, we have $k_m = d_1 + d_2$ and $k_m^S = d_2$.

It is natural to assume $\bm{Q}_{SC}=[\bm{\bar{Q}}_S,\bm{\bar{Q}}_S^{\perp}]\bm{B}^{(m)}\in\mathbb{R}^{n\times{k_m}}$, where $\bm{\bar{Q}}_S\in\mathbb{R}^{n\times{k_m^{S}}}$ is a matrix with columns being parts of orthonormal basis vectors for $\bm{H}_S$, and $\bm{\bar{Q}}_S^{\perp}\in\mathbb{R}^{n\times(k_m-k_m^{S})}$ is a matrix with columns composed of vectors in $\mathbb{P}_S^{\perp}\bm{H}_C$. and $\bm{B}^{(m)}\in\mathbb{R}^{k_m\times{k_m}}$ is a orthogonal matrix. Hence from (\ref{RankB}) we have
\begin{align}
  (\bm{Q}_{SC})^T\bm{P}_S^{\perp}\bm{Q}_{SC}&=\bm{I}_{k_m}-(\bm{Q}_{SC})^T\bm{P}_S\bm{Q}_{SC}\nonumber\\
  &=\bm{I}_{k_m}-(\bm{B}^{(m)})^T\left[\begin{array}{c}\bm{\bar{Q}}_S^T\\(\bm{\bar{Q}}_S^{\perp})^T\end{array}\right]\bm{P}_S[\bm{\bar{Q}}_S,\bm{\bar{Q}}_S^{\perp}]\bm{B}^{(m)}\nonumber\\
  &=\bm{I}_{k_m}-(\bm{B}^{(m)})^T\left[
  \begin{array}{cc}
  \bm{\bar{Q}}_S^T\bm{P}_S\bm{\bar{Q}}_S&\bm{\bar{Q}}_S^T\bm{P}_S\bm{\bar{Q}}_S^{\perp}\nonumber\\
  (\bm{\bar{Q}}_S^{\perp})^T\bm{P}_S\bm{\bar{Q}}_S&(\bm{\bar{Q}}_S^{\perp})^T\bm{P}_S\bm{\bar{Q}}_S^{\perp}
  \end{array}
  \right]\bm{B}^{(m)}
\end{align}
because
\begin{equation}
  \bm{P}_S\bm{\bar{Q}}_S=\bm{\bar{Q}}_S,\qquad\bm{P}_S\bm{\bar{Q}}_S^{\perp}=0,
\end{equation}
we have
\begin{align}
  \bm{\bar{Q}}_S^T\bm{P}_S\bm{\bar{Q}}_S&=\bm{I}_{k_m^{S}},\nonumber\\
  \bm{\bar{Q}}_S^T\bm{P}_S\bm{\bar{Q}}_S^{\perp}
  &=(\bm{\bar{Q}}_S^{\perp})^T\bm{P}_S\bm{\bar{Q}}_S
  =(\bm{\bar{Q}}_S^{\perp})^T\bm{P}_S\bm{\bar{Q}}_S^{\perp}=0,
\end{align}
therefore
\begin{align}
  (\bm{Q}_{SC})^T\bm{P}_S^{\perp}\bm{Q}_{SC}
  &=\bm{I}_{k_m}-(\bm{B}^{(m)})^T\left[
  \begin{array}{cc}\bm{I}_{k_m^{S}}&0\\0&0\end{array}
  \right]\bm{B}^{(m)}\\
  &=(\bm{B}^{(m)})^T\left[
  \begin{array}{cc}0&0\\0&\bm{I}_{k_m-k_m^{S}}\end{array}
  \right]\bm{B}^{(m)}
\end{align}
Let
\begin{equation}
  \bm{B}^{(m)}=\left[\begin{array}{c}\bm{B}_1^{(m)}\\\bm{B}_2^{(m)}\end{array}\right],
  \quad{\bm{B}_1^{(m)}\in\mathbb{R}^{k_m^S\times{k_m}},\ \bm{B}_2^{(m)}\in\mathbb{R}^{(k_m-k_m^S)\times{k_m}}},
\end{equation}
then we have
\begin{equation}
  (\bm{Q}_{SC})^T\bm{P}_S^{\perp}\bm{Q}_{SC}=(\bm{B}_2^{(m)})^T\bm{B}_2^{(m)},
\end{equation}
assume the singular value decomposition of $\bm{B}_2^{(m)}$ be
\begin{equation}
  \bm{B}_2^{(m)}=\bar{\bm{U}}[\bm{\Lambda},\bm{0}]\bar{\bm{V}},
\end{equation}
where $\bar{\bm{U}}\in\mathbb{R}^{(k_m-k_m^S)\times(k_m-k_m^S)}$ and $\bar{\bm{V}}\in\mathbb{R}^{k_m\times{k_m}}$ are orthogonal matrices and $\bm{\Lambda}$ is diagonal matrix. because of the orthogonality of $\bm{B}^{(m)}$,
\begin{equation}
  \bm{B}_2^{(m)}(\bm{B}_2^{(m)})^T=\bm{\bar{U}}\bm{\Lambda}^2\bm{\bar{U}}^T=\bm{I}_{k_m-k_m^S},
\end{equation}
hence we have $\bm{\Lambda}=\bm{I}_{k_m-k_m^S}$ and
\begin{equation}
  (\bm{B}_2^{(m)})^T\bm{B}_2^{(m)}=\bm{\bar{V}}\left[\begin{array}{cc}\bm{I}_{k_m-k_m^S}&\bm{0}\\\bm{0}&\bm{0}\end{array}\right]\bm{\bar{V}}^T,
\end{equation}
Then the following conclusion could be drawn: If target signal presents in sample data, then we have
\begin{equation}
  \Rank((\bm{Q}_{SC})^T\bm{P}_S^{\perp}\bm{Q}_{SC})=k_m-k_m^S,
\end{equation}
and all of its non-zero singular values (eigenvalues) are $1$.
\end{proof}

On the other hand, when there is no target but only clutter in the sample data, another lemma should hold.
\begin{Lemma}\label{lemma3}
   Under the hypothesis of non-existence of target signal in sample data, we have
  \begin{equation}
    \Rank((\bm{Q}_{SC})^T\bm{P}_S^{\perp}\bm{Q}_{SC})=k_m,
  \end{equation}
  and all of its nonzero singular values (eigenvalues) are $1$.
\end{Lemma}
\begin{proof}
It is noted that when the sample data $\bm{R}^{(m)}$ contains no target signal, We have $k_m=d_2$ when $m$ is sufficiently large. It is obvious that $\bm{Q}_{SC}$ is full rank, now
we should verify that $(\bm{Q}_{SC})^T\bm{P}_S^{\perp}\bm{Q}_{SC}$ is of full rank. That is,
\begin{equation}
  (\bm{Q}_{SC})^T
  \bm{P}_S^{\perp}\bm{Q}_{SC}\bm x= \bm 0\Longleftrightarrow{\bm x=\bm 0},
\end{equation}
In fact, we have
\begin{align}
  (\bm{Q}_{SC})^T\bm{P}_S^{\perp}\bm{Q}_{SC}\bm x=\bm 0&\Longleftrightarrow{\bm x^T(\bm{Q}_{SC})^T\bm{P}_S^{\perp}\bm{Q}_{SC}\bm x=0}\nonumber\\
  &\Longleftrightarrow{\bm x^T(\bm{Q}_{SC})^T\bm{P}_S^{\perp}\bm{P}_S^{\perp}\bm{Q}_{SC}\bm x=0}\qquad(\bm{P}_S^{\perp}\ \textrm{is idempotent matrix})\nonumber\\
  &\Longleftrightarrow\bm{P}_S^{\perp}\bm{Q}_{SC}\bm x=\bm 0\nonumber\\
  &\Longleftrightarrow\bm{Q}_{SC}\bm x\in\bm{H}_S\cap\bm{H}_C \quad (\text{because now} \Span(\bm Q_{SC}=\bm H_C))\nonumber\\
  &\Longleftrightarrow\bm{Q}_{SC}\bm x=\bm 0 \quad (\text{because} \dim(\bm H_S \cap \bm H_C)=0)\nonumber\\
  &\Longleftrightarrow{\bm x=\bm 0}\qquad(\bm{Q}_{SC}\ \textrm{is of full rank})
\end{align}
As for the proof about singular values, there will be a similar proof as the proof in Lemma 6, so we won't repeat the proof here.
\end{proof}
Now we continue the proof of Theorem \ref{H1Thm1N}.

According to (\ref{H1perturb}) in Lemma \ref{perturb} and (\ref{RankA}), we have
\begin{align}\label{OH1}
T(\bm{R}^{(m)})=&\det(\bm A + \bm E)\nonumber \\
\leq&\|\bm E\|_2^{k_m-k_m+k_m^S}\sum_{i=0}^{k_m-k_m^S}s_{k_m-k_m^S-i}(\bm{A})\|\bm E\|_2^i, \nonumber \\
=&s_{k_m-k_m^S}(\bm A)\|\bm E\|_2^{k_m^S}+O(\|\bm E\|_2^{k_m^S+1}),
\end{align}
where
\begin{eqnarray}\label{perturbE}
 \|\bm E\|_2 &=&  \| (\bm{Q}_{SC})^T \bm{P}_S^{\perp}\bm E^{(m)} + (\bm E^{(m)})^T\bm{P}_S^{\perp} \bm{Q}_{SC} +  (\bm E^{(m)})^T \bm{P}_S^{\perp}\bm E^{(m)}\|_2 \nonumber \\
& \leq & 2 \|(\bm{Q}_{SC})^T \bm{P}_S^{\perp}\bm E^{(m)}\|_2 + \|\bm{P}_S^{\perp}\bm E^{(m)}\|_2^2 \nonumber \\
&\leq & 2 \|\bm{Q}_{SC}\|_2 \|\bm{P}_S^{\perp}\bm E^{(m)}\|_2 + \|\bm{P}_S^{\perp}\bm E^{(m)}\|_2^2 \nonumber \\
&=& 2  \|\bm{P}_S^{\perp}\bm E^{(m)}\|_2 + \|\bm{P}_S^{\perp}\bm E^{(m)}\|_2^2 \nonumber \\
&\leq& 2  \|\bm{P}_S^{\perp}\bm E^{(m)}\|_F + \|\bm{P}_S^{\perp}\bm E^{(m)}\|_F^2.
\end{eqnarray}
\par Then, according to Lemma \ref{stoica} and Lemma \ref{LemmaRMF}, we have
\begin{equation}
\bm{P}_S^{\perp}\bm E^{(m)} \thicksim \mathcal{N}(0, \bm{P}_S^{\perp} \bm \Sigma_i^{(m)}(\bm{P}_S^{\perp})^T),
\end{equation}
and for any $\varepsilon >0$
\begin{equation}\label{ProbIneq}
\|\bm{P}_S^{\perp}\bm E^{(m)}\|_F^2 \leq (1+\varepsilon)\sum_{i=1}^{k_m} \Trace(\bm{P}_S^{\perp}\bm \Sigma_i(\bm{P}_S^{\perp})^T),
\end{equation}
holds with probability
\begin{equation}
\mathbb{P} \geq 1- \exp\{-\frac{k_m\cdot{n}\cdot\varepsilon^2}{C}\} . \nonumber
\end{equation}

Consider the right side of (\ref{ProbIneq}), we have
\begin{eqnarray}
\lefteqn{\sum_{i=1}^{k_m}  \Trace(\bm{P}_S^{\perp}\bm \Sigma_i\bm{P}_S) } \nonumber \\
&&= \sum_{i=1}^{k_m} \left(\frac{1}{m} \Big( \sum_{\stackrel{j=1}{j\neq i}}^{k_m} \frac{\lambda_i \lambda_j}{(\lambda_i-\lambda_j)^2} \Trace(\bm{P}_S^{\perp} \bm q_j \bm q_j^T \bm{P}_S^{\perp T}) + \sum_{j=k_m+1}^{n}\frac{\lambda_i\sigma^2}{(\sigma^2-\lambda_i)^2} \Trace( \bm{P}_S^{\perp} \bm q_j \bm q_j^T \bm{P}_S^{\perp T}) \Big)\right),
\end{eqnarray}
because
\begin{equation}
\Trace(\bm{P}_S^{\perp} \bm q_j \bm q_j^T \bm{P}_S^{\perp T}) = \Trace( \bm q_j^T \bm{P}_S^{\perp} \bm q_j) \leq 1,
\end{equation}
we have
\begin{equation}\label{TraceIneq}
\sum_{i=1}^{k_m}\Trace(\bm{P}_S^{\perp}\bm\Sigma_i\bm{P}_S)\leq\frac{1}{m}\left(\sum_{i=1}^{k_m} \sum_{\stackrel{j=1}{j\neq i}}^{k_m} \frac{\lambda_i \lambda_j}{(\lambda_i-\lambda_j)^2} +\sum_{i=1}^{k_m}(n-k_m)\frac{\lambda_i\sigma^2}{(\sigma^2-\lambda_i)^2}\right).
\end{equation}
Combine (\ref{TraceIneq}) and (\ref{ProbIneq}), we have for any $\varepsilon>0$ and $0 \leq \delta < 1$, if
\begin{equation}
\frac{1}{m}\left(  \sum_{i =1}^{k_m} \sum_{\stackrel{j=1}{j\neq i}}^{k_m} \frac{\lambda_i \lambda_j}{(\lambda_i-\lambda_j)^2}  + \sum_{i =1}^{k_m} (n-k_m)\frac{\lambda_i\sigma^2}{(\sigma^2-\lambda_i)^2} \right) \leq \frac{(\sqrt{\delta+1}-1)^2}{1+\varepsilon},
\end{equation}
or equivalently,
\begin{equation}
m \geq\frac{1+\varepsilon}{(\sqrt{\delta+1}-1)^2}\left( \sum_{i =1}^{k_m} \sum_{\stackrel{j=1}{j\neq i}}^{k_m} \frac{\lambda_i \lambda_j}{(\lambda_i-\lambda_j)^2}  + \sum_{i =1}^{k_m} (n-k_m)\frac{\lambda_i\sigma^2}{(\sigma^2-\lambda_i)^2} \right),
\end{equation}
then
\begin{equation}\label{concent1}
\|\bm{P}_S^{\perp}\bm E^{(m)}\|_F^2 \leq (\sqrt{\delta+1}-1)^2,
\end{equation}
holds with probability
\begin{equation}
\mathbb{P} \geq 1- \exp\{-\frac{k_m\cdot{n}\cdot\varepsilon^2}{C}\} . \nonumber
\end{equation}
Then combining (\ref{perturbE}) with (\ref{concent1}), we get
\begin{equation}
\|\bm E\|_2\leq 2(\sqrt{\delta+1}-1)+(\sqrt{\delta+1}-1)^2={\delta},
\end{equation}
thus we have
\begin{equation}
T(\bm R^{(m)})\leq s_{k_m-k_m^S}((\bm{Q}_{SC})^T\bm{P}_S^{\perp}\bm{Q}_{SC})\delta^{k_m^S} + O(\delta^{k_m^S+1}),
\end{equation}
holds with overwhelming probability.

Furthermore, according to the definition of elementary symmetric function of singular values in (\ref{ESF}) and (\ref{Rank1}), it can be verified easily that
\begin{equation}
  s_{k_m-k_m^S}((\bm{Q}_{SC})^T\bm{P}_S^{\perp}\bm{Q}_{SC})=1,
\end{equation}
hence we have
\begin{equation}
  T(\bm R^{(m)})\leq \delta^{k_m^S} + O(\delta^{k_m^S+1}),
\end{equation}

If the number $m$ of sample data is large sufficiently, then we have
\begin{equation}
  k_m=d_1+d_2,\qquad{k_m^S=d_2},
\end{equation}
therefore
\begin{equation}
  \frac{1}{|T^2(\bm R^{(m)})|}\leq \delta^{d_2} + O(\delta^{d_2+1}),
\end{equation}

On the other hand, using (\ref{H0perturb}) in lemma \ref{perturb}, we can similarly obtain the corresponding result for non-target scenario.
\begin{equation}\label{OH0}
|\det(\bm A+\bm E)-\det(\bm A)| \leq  s_{k_m-1}(\bm A)\|\bm E\|_2 + O(\|\bm E\|_2^2),
\end{equation}
When $m$ is large sufficiently, we have $k_m=d_1$ and $\bm{Q}_{SC}=\bm{Q}_C$. According to (\ref{detperturb}), we have
\begin{equation}
\det((\bm{Q}_{SC})^T\bm{P}_S^{\perp}\bm{Q}_{SC})=\Vol_{d_1+d_2}^2([\bm{Q}_S,\bm{Q}_C]) := \tau^2(\bm{H}_S,\bm{H}_C),
\end{equation}
Similar to discussion above, for any $0 \leq \delta <1$ and $\varepsilon >0$, if
\begin{equation}
m \geq\frac{1+\varepsilon}{(\sqrt{\delta+1}-1)^2}\left( \Big( \sum_{i =1}^{d_1} \sum_{\stackrel{j=1}{j\neq i}}^{d_1} \frac{\lambda_i \lambda_j}{(\lambda_i-\lambda_j)^2}+\sum_{i =1}^{d_1}(n-d_1)\frac{\lambda_i\sigma^2}{(\sigma^2-\lambda_i)^2}  \Big)\right),
\end{equation}
then
\begin{equation}
\|\bm{P}_S^{\perp}\bm E^{(m)}\|_F^2 \leq (\sqrt{\delta+1}-1)^2,
\end{equation}
hold with probability
\begin{equation}
\mathbb{P} \geq 1-\exp\{-\frac{d_1\cdot{n}\cdot \varepsilon^2}{C}\}
\end{equation}
therefore
\begin{equation}
|T(\bm R^{(m)})-\tau^2(\bm{H}_S,\bm{H}_C)|\leq s_{d_1-1}(\bm{Q}_C^T\bm{P}_S^{\perp}\bm{Q}_C)\delta + O(\delta^{2}).
\end{equation}

\bibliography{IEEEabrv,CompressiveSensingMultiuserDetection}
\bibliographystyle {IEEEtran}
\end{document}